\documentclass[journal]{IEEEtran}

\usepackage{graphicx}
\usepackage{amsmath}
\usepackage{amssymb}
\usepackage{cite}
\usepackage[colorlinks=Flase]{hyperref}
\usepackage{bbm}
\usepackage{algorithm} 
\usepackage{algpseudocode}
\usepackage{caption}
\usepackage{steinmetz}
\usepackage{subcaption}

\usepackage{amsthm}
\usepackage[english]{babel}
\usepackage{float}
\newcounter{relctr} %% <- counter for relations
\everydisplay\expandafter{\the\everydisplay\setcounter{relctr}{0}} %% <- reset every eq
\DeclareMathOperator*{\argmin}{argmin} % no space, limits underneath in displays
\DeclareMathOperator*{\argmax}{argmax} % no space, limits underneath in displays
\DeclareMathOperator*{\minimize}{minimize} % no space, limits underneath in displays
\DeclareMathOperator*{\maximize}{maximize} % no space, limits underneath in displays
\newtheorem{theorem}{Theorem}

\usepackage[figurename=Fig.]{caption}

\newcommand\labelrel[2]{%
  \begingroup
    \refstepcounter{relctr}%
    \stackrel{\textnormal{(\alph{relctr})}}{\mathstrut{#1}}%
    \originallabel{#2}%
  \endgroup
}
\AtBeginDocument{\let\originallabel\label} %% <- store original definition
\ifCLASSINFOpdf
\else
\fi
\usepackage[normalem]{ulem}
\usepackage{xcolor}

\newtheoremstyle{remarkstyle}%
  {}%                       % Space above
  {}%                       % Space below
  {\itshape}%               % Body font
  {}%                       % Indent amount
  {\itshape}%              % Theorem head font
  {.}%                      % Punctuation after theorem head
  {.5em}%                   % Space after theorem head
  {}%                       % Theorem head spec (can be left empty, meaning `normal`)

\theoremstyle{remarkstyle}

\makeatletter
\newcommand{\algorithmicprintspacing}{\renewcommand{\baselinestretch}{0.9}\normalsize}

\makeatother

\begin{document}

\title{Waveform Optimization and Beam Focusing for Near-field Wireless Power Transfer with Dynamic Metasurface Antennas and Non-linear Energy Harvesters}

\author{Amirhossein~Azarbahram,~\IEEEmembership{Graduate~Student~Member,~IEEE,}
        Onel~L.~A.~López,~\IEEEmembership{Senior~Member,~IEEE},
        and~Matti~Latva-Aho,~\IEEEmembership{Fellow,~IEEE} % <-this stops a space
\thanks{A. Azarbahram, O. L\'opez and M. Latva-Aho are with Centre for Wireless Communications (CWC), University of Oulu, Finland, (e-mail: \{amirhossein.azarbahram, onel.alcarazlopez, matti.latva-aho\}@oulu.fi).}%
\thanks{This work is partially supported in Finland by the Finnish Foundation for Technology Promotion, Academy of Finland (Grants 348515 and 346208 (6G Flagship)), the Finnish-American Research and Innovation Accelerator, and by the European Commission through the Horizon Europe/JU SNS project Hexa-X-II (Grant Agreement no. 101095759).}}

% \markboth{Journal of \LaTeX\ Class Files,~Vol.~14, No.~8, March~2023}%
% {Shell \MakeLowercase{\textit{et al.}}: Bare Demo of IEEEtran.cls for IEEE Journals}

\maketitle

\begin{abstract}
Radio frequency (RF) wireless power transfer (WPT) is a promising technology for future wireless systems. However, the low power transfer efficiency (PTE) is a critical challenge for practical implementations. One of the main inefficiency sources is the power consumption and loss introduced by key components such as high-power amplifier (HPA) and rectenna, thus they must be carefully considered for PTE optimization. Herein, we consider a near-field RF-WPT system with a dynamic metasurface antenna (DMA) at the transmitter and non-linear energy harvesters. We provide a mathematical framework to calculate the power consumption and harvested power from multi-tone signal transmissions. Based on this, we propose an approach relying on alternating optimization and successive convex approximation for waveform optimization and beam focusing to minimize power consumption while meeting energy harvesting requirements. Numerical results show that increasing the number of transmit tones reduces the power consumption by leveraging the rectifier's non-linearity more efficiently. Moreover, they demonstrate that increasing the antenna length improves the performance, while DMA outperforms fully-digital architecture in terms of power consumption. Finally, our results verify that the transmitter focuses the energy on receivers located in the near-field, while energy beams are formed in the receivers' direction in the far-field region.

\end{abstract}


\IEEEpeerreviewmaketitle

\section{Introduction}

\IEEEPARstart{F}{uture} wireless systems will facilitate efficient and eco-friendly communication across a myriad of low-power devices, fostering a sustainable society. Achieving this requires uninterrupted connectivity among these devices and with the underlying infrastructure, all while mitigating disruptions arising from battery depletion \cite{intro1,intro2, lópez2023highpower}. This is potentially facilitated by energy harvesting (EH) technologies providing wireless charging, thus, easing the maintenance of Internet of Things (IoT) devices and increasing their lifespan. Moreover, EH may lead to improved energy efficiency and reduced emission footprints across the network \cite{intro3}. 

EH receivers may harvest energy from two types of sources: those that exist in the surrounding environment, and those that are specifically designated for energy transmission. From a transmission perspective, the latter is supported by wireless power transfer (WPT) technologies, e.g., based on inductive coupling, magnetic resonance coupling, laser power beaming, and radio frequency (RF) radiation. Among them, RF-WPT is promising for charging multiple receivers relatively far from the transmitter by exploiting the broadcast nature of wireless channels. Furthermore, this can be accomplished over the same infrastructure used for wireless communications. Notably, the most important challenge toward maturing RF-WPT is related to increasing the inherently low end-to-end power transfer efficiency (PTE) \cite{intro3}. Herein, we focus on RF-WPT, which is referred to as WPT in the following.

\subsection{Preliminaries}

\begin{figure}[t]
    \centering
    \includegraphics[width=\columnwidth]{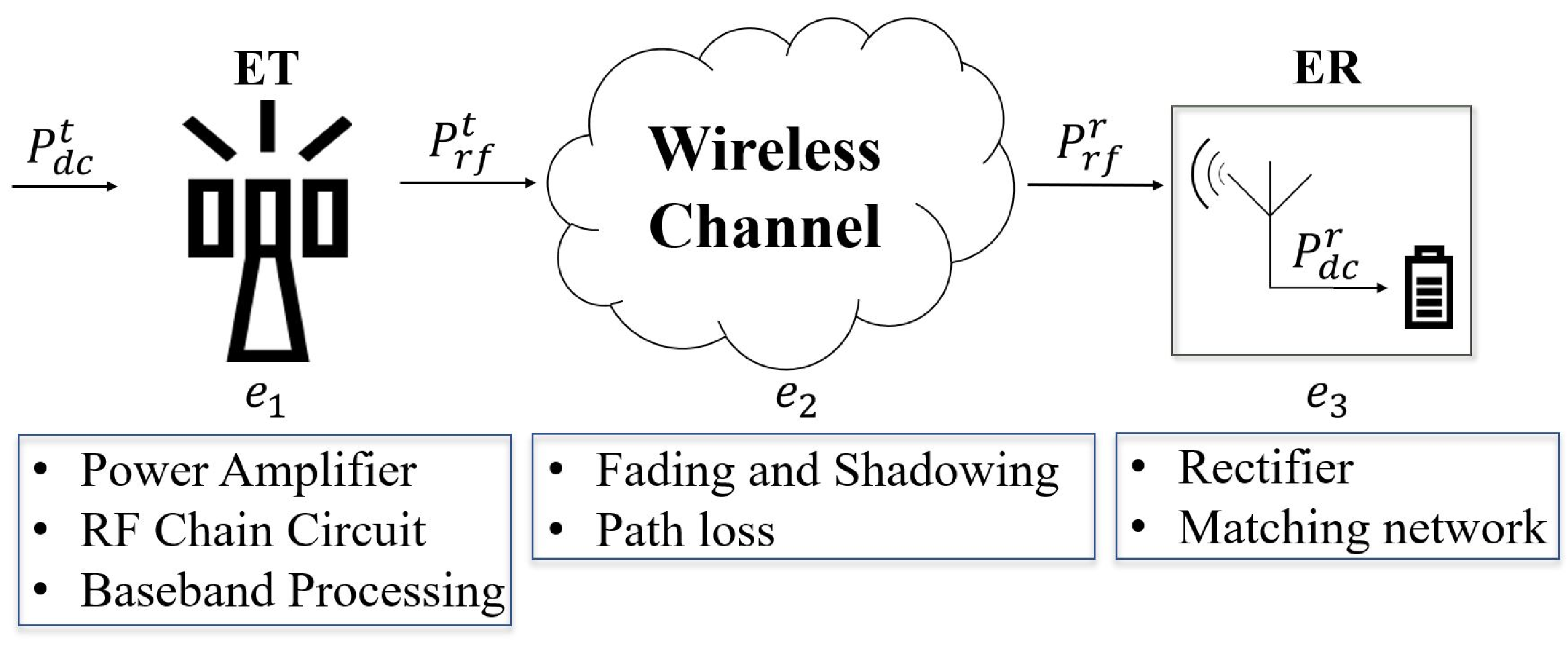}
    \caption{Block diagram of a typical WPT system. The power consumption and loss sources are listed under each block.}
    \label{fig:wptbase}
\end{figure}

The end-to-end PTE depends on the performance of the key building blocks of a WPT system, namely energy transmitter (ET), wireless channel, and energy receiver (ER) as illustrated in Fig.~\ref{fig:wptbase}. At first, a signal is generated and amplified using a direct current (DC) power source at the ET. Then, it is upconverted to the RF domain and transmitted over the wireless channel. Finally, the ER converts the received RF signal to DC for EH purposes. Indeed, the end-to-end PTE comprises: DC-to-RF, RF-to-RF, and RF-to-DC power conversion efficiency, i.e.,
\begin{equation}\label{eq:wptbase}
    e = \underbrace{\frac{P^{t}_{rf}}{P^{t}_{dc}}}_{e_1} \times \underbrace{\frac{P^{r}_{rf}}{P^{t}_{rf}}}_{e_2} \times \underbrace{\frac{P^{r}_{dc}}{P^{r}_{rf}}}_{e_3} = \frac{P^{r}_{dc}}{P^{t}_{dc}}.
\end{equation}
In WPT, both transmitter and ER introduce non-linearities to the signals affecting the amount of harvested power. Indeed, an appropriately designed transmit signal leveraging these non-linearities may reduce the power consumption at the transmitter and/or increase the RF-to-DC conversion efficiency at the ER \cite{clerckx2017communicationsWPTbase, clreckxWFdesign}. Specifically, using multiple transmit frequency tones can lead to high peak-to-average power ratio (PAPR) signals at the ER, enhancing the rectifier RF-to-DC conversion efficiency \cite{papr2}. Meanwhile, the RF signal can be focused towards the ER using energy beamforming (EB), which affects the transmit/receive waveform, to cope with the channel inefficiencies captured by $e_2$, thereby enhancing the amount of RF power that can be harvested \cite{intro3}. Notice that the active transmit components consume power, while the passive elements introduce power losses, both of which impact $e_1$ and must be considered. Therefore, $e_1$, $e_2$, and $e_3$ are correlated, suggesting that their joint optimization may lead to significant gains in terms of end-to-end PTE.

The end-to-end PTE is also affected by the system power consumption. One of the main factors contributing greatly to the power consumption is the transmitter's architecture, which also determines the beamforming approach. For example, in a fully-digital structure, each antenna element necessitates a dedicated RF chain with its corresponding high-power amplifier (HPA), consuming a significant amount of power. An HPA aims to amplify the input signal to compensate for the path loss and fading in a wireless system. Moreover, the signal amplification in the HPA requires a DC power source, which accounts for the majority of the power consumption. The main drawback of the fully-digital architecture is the high complexity and cost, making it impractical for applications requiring massive multiple-input multiple-output (MIMO) implementations. Alternatively, analog architectures using, e.g., passive phase shifters, are cheaper but offer limited degrees of freedom for EB. Thus, a hybrid architecture implementation combining both approaches is often more appealing in practice. Hybrid architectures offer a trade-off between complexity (cost) and beamforming flexibility \cite{hybridbeamsurvey, hybrid2}. 

Although hybrid beamforming using phase shifters promotes cost reduction, it still requires complex analog networks for phase shifting. Moreover, the complexity of traditional phase-shifter-based hybrid architectures scales with the complexity of this analog network. For instance, a fully-connected network of phase shifters for connecting the output of the RF chains to the antenna elements is the most complex architecture, while there are some other low-complexity architectures, e.g., dynamic array-of-subarrays architecture \cite{arraysubarray1, arraysubarray2}. There are emerging technologies to provide hybrid beamforming capability with an even lower cost and complexity, e.g., reconfigurable intelligent surface (RIS)-aided systems \cite{RIS-basis} and dynamic metasurface antennas (DMAs) \cite{DMAbase}. Notice that, RIS is an assisting node that provides the passive beamforming capability using reflective elements,\footnote{There are also some novel RIS-based transmit structures that leverage RIS with active elements to provide extra capabilities \cite{active_ris}.} while DMA is a transceiver consisting of configurable metamaterial elements and a limited number of RF chains. DMA avoids analog network implementation challenges and provides hybrid beamforming capability with low cost and complexity. Each of these architectures may be preferred based on the system setup. For instance, employing multiple low-cost RISs helps to cover the blind spots that are prone to weak signal reception in a large area. However, reflecting surfaces in RIS-assisted systems lack baseband processing capability to perform channel estimation and send pilot signals. Thus, acquiring accurate enough channel state information (CSI) to attain a suitable passive beamforming gain might force huge overheads to the system \cite{RIS-Challenges}. On the other hand, DMA is a transceiver and has sufficient baseband processing capability for channel estimation. However, its implementation requires some RF chains making it more costly than RIS for large-scale implementation. All in all, both of these architectures support low-cost transmitter deployments, while the choice highly depends on the considered system setup. Interestingly, the authors in \cite{RIS-DMA} utilize a system model comprising both RIS and DMA structures for uplink MIMO communication while assuming that the channel is perfectly known.

\subsection{Prior works}

There are many works either focusing on EB, waveform optimization, or joint waveform and beamforming design for fully-digital WPT systems. The authors in \cite{OnelRadioStripes} utilize EB to power multiple devices in a MIMO system consisting of radio stripes, while the deployment problem of this transmit architecture is investigated in \cite{azarbahram2023radio}. Furthermore, a low-complexity beamforming design relying on the statistics of the channel is proposed in \cite{onellowcomp} to fairly power a set of single-antenna devices. In the mentioned works, none of the practical system non-linearities are considered and the focus is on the received RF power at the devices. In \cite{BFRFDCsingletone}, transmit beamforming and RF and DC combining at the ER are leveraged to increase the received DC power in a MIMO WPT system. Although this work considers the rectifier's non-linearity, the waveform design is not investigated. Moreover, a low-complexity waveform design for single-user setups is proposed in \cite{brunolowcomp}, while the large-scale multi-antenna WPT scenario is addressed in \cite{largescale}. The authors in \cite{jointWFandBFMIMOclreckx} leverage beamforming and a multi-sine waveform in a MIMO WPT system to enhance the harvested power. Notably, the frameworks in \cite{brunolowcomp, largescale, jointWFandBFMIMOclreckx} consider the rectifier non-linearity. Interestingly, the authors in \cite{SISOAllClreckx} perform waveform and beamforming optimization while considering both main non-linearity sources (i.e., HPA and rectifier) aiming to maximize the harvested DC power in a WPT system. 

Although most of the works on WPT systems in the literature focus on a traditional fully-digital architecture, novel low-cost transmitters have also attained significant attention recently. For instance, DMA is utilized in \cite{DMAWPT, MyEBDMA} for a near-field WPT system, while the authors in \cite{hybrid_minpower_SWIPT} propose a minimum-power beamforming design for meeting quality of service requirements of the ERs in a simultaneous wireless information and power transfer (SWIPT) system. However, none of these studies have considered the rectifier non-linearity and its impact on the harvested DC power. Notably, the joint waveform and beamforming design problem in RIS-aided WPT and SWIPT systems is investigated in \cite{clerckxRISWPT} and \cite{RIS-SWIPT-Clreckx}, respectively. Furthermore, the two latter works consider the EH non-linearity at the ER side, thus, taking into account its impact on the harvested DC power.

\subsection{Contributions}

All in all, WPT systems have received considerable attention for some time. Still, more effort is needed to reduce the system power consumption, thus increasing the end-to-end PTE. For this, low-cost multi-antenna transmitters like DMA are appealing and may pave the way for charging devices efficiently in massive IoT deployments. Moreover, multiple studies aimed to enhance the amount of harvested DC power (with rectifier non-linearity) in far-field WPT systems or received RF power (without rectifier non-linearity) in near-field WPT. As mentioned before, the ER does not perform RF-to-DC conversion linearly and a proper waveform design must leverage the ER's RF-to-DC conversion efficiency dynamics. Thus, it is imperative to take into account the impact of the rectifier's non-linearity. To the best of our knowledge, no work has yet investigated the power consumption of a multi-antenna WPT system for meeting the EH requirements of a multi-user setup while considering the ER non-linearity, especially when using low-cost transmitters. Herein, we aim precisely to fill this research gap. Our main contributions are as follows: \\
1) We formulate a joint waveform optimization and beam focusing problem for a multi-user multi-antenna WPT system with both a fully-digital and a DMA architecture. Due to the huge potential of near-field WPT systems for future practical WPT applications \cite{lópez2023highpower}, we present our system model relying on a near-field wireless channel, which can inherently capture far-field conditions as well. Notice that there are some previous works \cite{clerckxRISWPT, SISOAllClreckx} focused on increasing the amount of harvested power in far-field WPT systems while considering the ER non-linearity. However, the literature lacks a minimum-power waveform and beamforming design (even for far-field channels) that can deal with meeting EH requirements. This is a critical gap to fill since such formulation mimics a practical setup where the EH receivers inform their DC power demands and the WPT system must serve them with minimum power consumption, thus maximum end-to-end PTE. Since the HPA is an active element incurring most of the power consumption at the transmitter side, we model the power consumption of a class-B HPA as a function of its output power. Notably, our problem for fully-digital architecture shares some similarities with the one discussed in \cite{SISOAllClreckx} as objectives and constraints are interchanged. However, there are two main differences. First, since the problem is highly non-linear and non-convex, aiming for a minimum-power waveform and beamforming design leads to a different and complex problem due to the appearance of the highly non-convex constraints and a non-convex power consumption term relying on the transmit waveform. Second, the main focus of our work here is on a DMA-assisted system, which introduces much more complexity to the problem due to the coupling between the optimization variables, i.e., the frequency response of the elements and the digital beamforming complex weights, and the Lorentzian-constrained phase response of the elements. Note that the phase shift introduced by the metamaterial elements is correlated with their amplitude, which results in a different beamforming problem than other architectures, e.g., RIS-assisted systems \cite{RIS-SWIPT-Clreckx} and phase shifter-based hybrid beamforming \cite{hybrid2}. Mathematically, when dealing with those latter architectures, both RIS passive elements or phase shifters introduce a phase shift to the signal with constant loss,\footnote{Note that in most of the works in the literature \cite{RIS-SWIPT-Clreckx, hybrid_minpower_SWIPT}, without loss of generality, this phase shifting process by the analog network or RIS elements is considered to be lossless.} while each phase-shifting configuration in DMA elements leads to a different propagation loss introduced to the transmit signal. Meanwhile, observe that our utilized DMA model is similar to \cite{DMAWPT, near-field}, which is the standard model in the literature \cite{DMAbase}. However, we investigate this model for WPT when considering ER non-linearity, which leads to a different problem in terms of complexity and required solutions, especially when aiming for minimum-power waveform and beamforming design instead of aiming for received power maximization. All in all, the complexities associated with our specific problem make the existing optimization frameworks for WPT systems inapplicable to our system, calling for novel approaches.

    %\item 
2) We propose a method relying on alternating optimization and successive convex approximation (SCA) to efficiently solve the waveform and beamforming optimization problem in the DMA-assisted WPT system. Specifically, we decouple the optimization problem to maximize the minimum received DC power by tuning the frequency response of the metamaterial elements, while minimizing the power consumption needed for meeting the EH requirements when optimizing the digital precoders. Generally, a huge complexity is introduced to the waveform optimization problems by time sampling since the number of samples should be relatively large to result in a reliable framework \cite{SISOAllClreckx, clreckxWFdesign, jointWFandBFMIMOclreckx}. To cope with this, we reformulate the received DC power of the ERs based on the spectrum of the received waveform, which removes the time dependency in the problem. Then, the metamaterial elements and the digital precoders are alternatively optimized using SCA. Motivated by the influence of variable initialization on the SCA performance, we propose a low-complexity initialization algorithm for the digital precoders and DMA weights by leveraging the channel characteristics and dedicating RF signals to the ERs. Furthermore, the complexity of the proposed optimization framework scales with the number of ERs, antenna elements, and frequency tones. Note that the location of the ERs is considered to be perfectly known at the transmitter side, leading to perfect CSI at the transmitter. Indeed, the proposed framework can be used for any channel model as long as the CSI is known. However, the focus of the work is on the radiative near-field region for WPT applications; thus, we also use the term beam focusing in addition to traditional beamforming, which is typical for far-field communications.
    
3) We illustrate the convergence of the proposed optimization method numerically and show that the complexity increases with the antenna length, number of transmit tones, and number of receivers. Furthermore, we provide evidence that increasing the antenna length and the number of transmit tones reduces power consumption, while it increases with the number of receivers and ER distance. Moreover, our findings evince that the DMA-assisted architecture outperforms the fully-digital counterpart in terms of power consumption in our system. Moreover, the performance gap between these architectures is shown to depend on the specific HPAs' saturation power, number of devices, and ER distance. Additionally, we verify by simulation that the transmitter can accurately focus the energy on the ER location in the near-field region, while energy beams are only formed toward specific directions in the far-field region.

The remainder of the paper is structured as follows. Section~\ref{TX and consum} introduces the system model, including the transmit architectures, and signal and power consumption modeling. The optimization problem for a joint waveform and beamforming design, together with the proposed solving approach, are elaborated in Section~\ref{Optimization}. Section~\ref{sec:init} discusses the proposed initialization algorithm, Section~\ref{result} presents the numerical results, while Section~\ref{conclusion} concludes the paper.

\textbf{Notations:} Bold lower-case letters represent column vectors, while non-boldface characters refer to scalars, ${\mathbf{a}}\odot{\mathbf{b}}$ denotes the Hadamard product of $\mathbf{a}$ and $\mathbf{b}$, and $\{x\}$ is a set that contains $x$. The $l_2$-norm of a vector is denoted as $| \cdot |$. The average w.r.t. $x$ is represented by $\mathbb{E}_x$ and $(\cdot)^T$ and $(\cdot)^\star$ are used to indicate the transpose and conjugate of a matrix or vector, respectively. Furthermore, the real and the imaginary parts of a complex number are denoted by $\Re \{\cdot\}$ and $\Im \{\cdot\}$, respectively. Additionally, $\lfloor{\cdot}\rfloor$ is the floor operator, and $\langle{\cdot}\rangle$ denotes the phase of a complex number. 

\section{System Model}\label{TX and consum}

\begin{figure*}[t]
    \centering
    \includegraphics[width=0.4\textwidth]{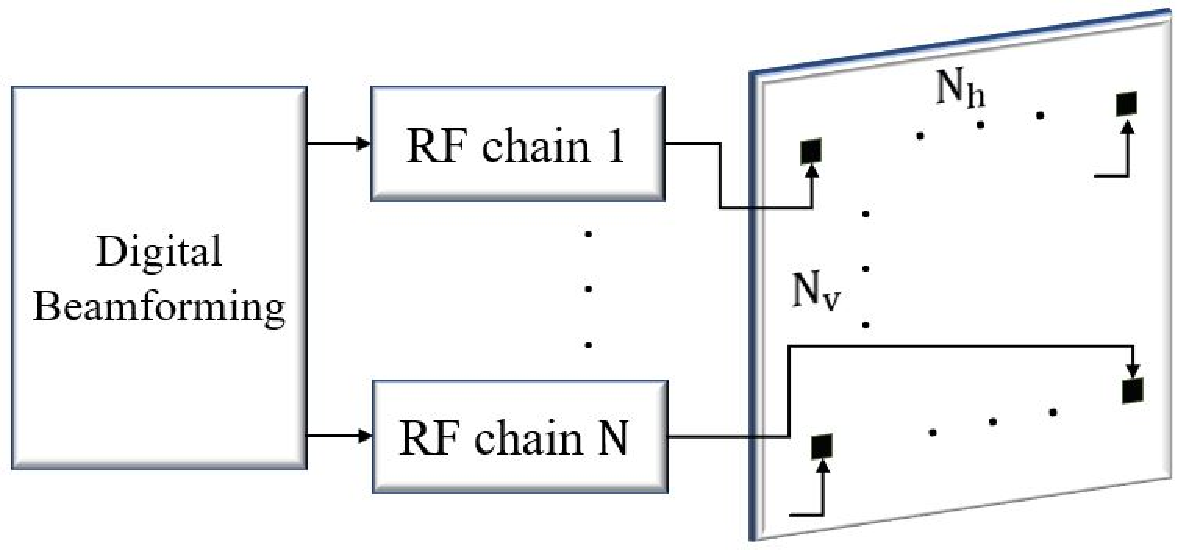}
    \label{fig:fd}
    \hspace{5mm}
    \includegraphics[width=0.5\textwidth]{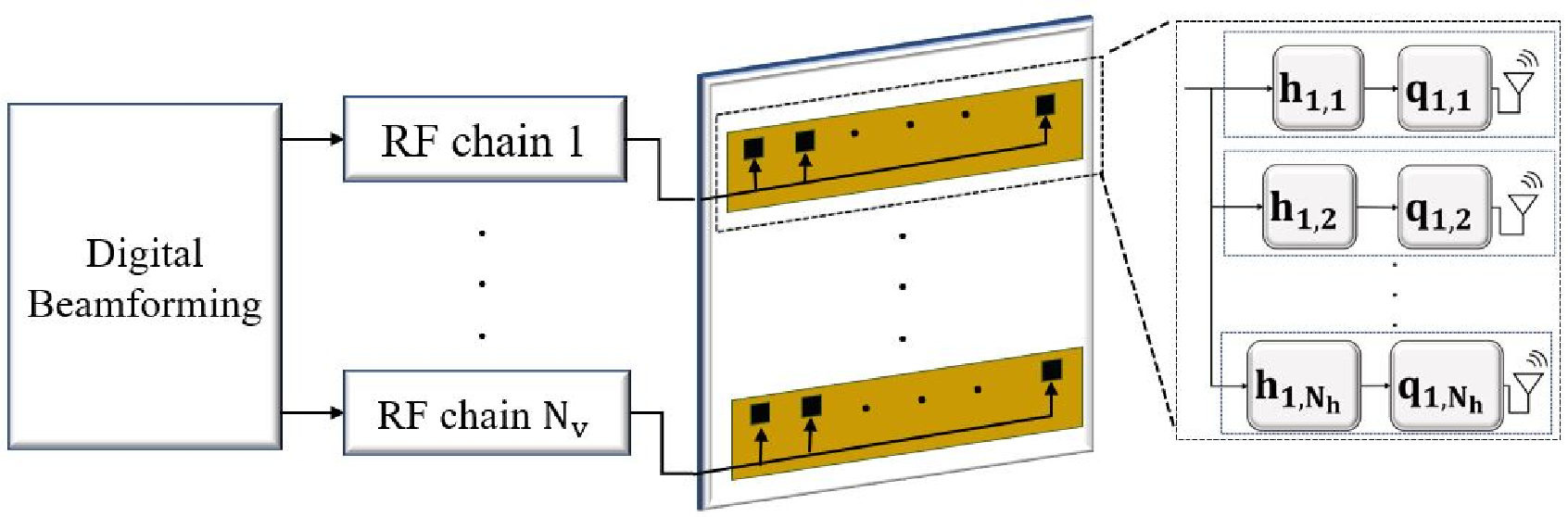}
    \label{fig:dma}
    \caption{Transmit antenna architectures. (a) fully-digital architecture (left) and (b) DMA-assisted architecture (right).}
    \label{fig:strcut}
\end{figure*}

We consider a multi-antenna WPT system to charge $M$ single-antenna EH devices. The received RF power at the ER is transferred to the rectifier input using a matching network. Then, it is converted to DC power by the rectifier, while $\bar{P}_{m}$ denotes the EH requirement of the $m$th ER. 

As previously mentioned, multi-tone waveforms can be exploited to leverage the rectifier non-linearity and achieve a better end-to-end PTE. Hence, we consider multi-tone signals with $N_f$ tones at frequencies $f_1, f_2, \cdots, f_{N_f}$ for power transmission purposes. Without loss of generality, we set $f_n = f_1 + (n-1) \Delta_f,\quad n = 1, \ldots, N_f$,
where $f_1$ and $\Delta_f$ are the lowest sub-carrier frequency and the sub-carrier spacing, respectively. 

\subsection{Transmit Antenna Architectures}
The transmitter is equipped with a uniform planar array (UPA) and $N_{rf} \geq M$ RF chains. The radiating elements are spaced uniformly, with $N_h$ and $N_v$ being the number of elements in the horizontal and vertical direction, respectively. Thus, the total number of elements is $N = N_v \times N_h$. Two types of transmit antenna architectures are considered:
   
1) Fully-digital architecture, which requires a dedicated RF chain for each radiating element, thus $N_{rf} = N$, as shown in Fig.~\ref{fig:strcut}a. In a fully-digital architecture, there is a single-stage beamforming process. Despite the high deployment cost and complexity, a fully-digital structure offers the highest number of degrees of freedom in beamforming. 

2) DMA-assisted architecture, which comprises $N_v$ waveguides, each fed by a dedicated RF chain and composed of $N_h$ configurable metamaterial elements. Therefore, the number of RF chains, and consequently the cost and complexity, is considerably reduced compared to digital structures, making DMA suitable for massive MIMO applications. Notice that DMA-assisted systems employ a two-stage beamforming process, i.e., digital beamforming, followed by the tuning of the amplitude/phase variations introduced by the metamaterial elements, as illustrated in Fig.~\ref{fig:strcut}b. Herein, $q_{i, l}$ is the tunable frequency response of the $l$th metamaterial element in the $i$th waveguide, while $h_{i,l}$ is the corresponding waveguide propagation loss, which will be explained in detail later.    

\subsection{Channel Model}

In wireless communications, the region where the receivers are located between the Fraunhofer and Fresnel distances denoted respectively as $d_{fr}$ and $d_{fs}$, is the radiative near-field region, which is referred to as the near-field region in the following. Specifically, a receiver at distance $d$ from a transmitter experiences near-field conditions if 
\begin{equation}
    \sqrt[3]{\frac{D^4}{8\lambda_1}} = d_{fs} < d < d_{fr} = \frac{2D^2}{\lambda_1}, \nonumber
\end{equation}
where $D$ is the antenna diameter, i.e., the largest size of the antenna aperture, $\lambda_1 = c/{f_1}$ is the corresponding wavelength to the system operating frequency, and $c$ is the speed of light. Notice that both system frequency and antenna form factor influence the region of operation. Therefore, by moving toward higher frequencies, e.g., millimeter wave and sub-THz bands, and/or utilizing larger antenna arrays, the far-field communication assumption regarding planar wavefronts may not be valid anymore. Instead, wavefronts impinging a receive node may be strictly spherical, thus, with advanced capabilities to focus the transmit signals on specific spatial points rather on spatial directions. Thus, we use the term beam focusing since we consider the system to operate in the radiative near-field region.

Notice that one of the main applications of WPT systems is in indoor environments with line-of-sight and near-field communication, e.g., restaurants, warehouses, and shopping malls. Thus, we employ the near-field line-of-sight channel model described in \cite{near-field}. The Cartesian coordinate of the $l$th radiating element in the $i$th row is $\mathbf{g}_{i,l} = [x_{i,l}, y_{i,l}, z_{i,l}]^T$. Additionally, $i = 1,2,\ldots,N_v$ and $l = 1,2,\ldots,N_h$. The channel coefficient between ER $m$ and the $l$th element in the $i$th row at the $n$th sub-carrier is given by
\begin{equation}\label{eq:channelcoef}
    \gamma_{i,l,m,n} = A_{i,l,m,n} e^{\frac{-j2\pi}{\lambda_n} d_{i,l,m}},
\end{equation}
where ${\frac{2\pi}{\lambda_n} d_{i,l,m}}$ is the phase shift caused by the propagation distance of the $n$th tone, with wavelength $\lambda_n$, and $d_{i,l,m} = |\mathbf{g}_m - \mathbf{g}_{i,l}|$ is the Euclidean distance between the element and the ER located at $\mathbf{g}_m$. Moreover,
\begin{equation}
    A_{i,l,m,n} = \sqrt{F(\Theta_{i,l,m})}\frac{\lambda_n}{4\pi d_{i,l,m}}
\end{equation}
is the corresponding channel gain coefficient. Here, $\Theta_{i,l,m} = (\theta_{i,l,m},\psi_{i,l,m})$ is the elevation-azimuth angle pair, and $F(\Theta_{i,l,m})$ is the radiation profile of each element. In addition, we employ the radiation profile as presented in \cite{anetnna_radiation}, where
\begin{equation}\label{eq:radiation}
        F(\Theta_{i,l,m}) = \begin{cases}
        G_t\cos{(\theta_{i,l,m})}^{\frac{G_t}{2}-1}, & \theta_{i,l,m} \in [0,\pi/2],
        \\
        0, & \text{otherwise},
        \end{cases}
\end{equation}
$G_t = 2(b+1)$ is the transmit antenna gain, and $b$ denotes the boresight gain, which depends on the specifications of the antenna elements. Note that the channel coefficient becomes $A_{m} e^{-j\psi_{i,l,m}}$ for far-field communication, where $A_m$ only depends on the distance of the receiver $m$ from the transmitter and $\psi_{i,l,m}$ is solely determined by the ER direction and the relative disposition of the antenna elements within the array. Notably, the radiation profile might be more complicated in practice, e.g., due to back lobes. However, this work considers perfect CSI at the transmitter side; thus, the proposed frameworks apply also to other radiation profiles and channel models and we adapt a basic ideal radiation profile without loss of generality.

\subsection{Transmit \& Receive Signals}\label{Section:Sys:TxRXSig}

The signal at the input of the $i$th HPA is given by
\begin{equation}\label{eq:basicform}
    x_i(t) =  \sum_{n = 1}^{N_f} {\omega_{i, n} e^{j2\pi f_n t}}, \quad i = 1, \ldots, N_{rf},
\end{equation}
where $\omega_{i, n}$ is the complex weight of the $n$th frequency tone of the $i$th muti-tone waveform.
 The HPA introduces signal distortion and models such as the Rapp model \cite{rapp1991effects} capture this non-linearity. It is shown in \cite{SISOAllClreckx} that when the HPA operates in the non-linear regime, choosing a single-carrier is preferred to a multi-carrier waveform. The reason is that a single-carrier waveform is less deteriorated by the adverse effect of the signal distortion caused by the HPA when operating in the non-linear regime. On the other hand, when HPAs operate in the linear regime, thus, not causing amplitude and phase distortion to the signal, multi-carrier waveforms are preferred since they leverage the rectifier's non-linearity and enhance the harvested power performance. 

Note that the non-linear regime of the HPA happens near the saturation voltage. Thus, in practice, the HPAs can be properly chosen to have a suitable value of the saturation voltage based on the system setup and avoid operation in the non-linear regime. Since the aim of this work is to design multi-carrier waveforms for DMA-assisted WPT, we consider HPAs to operate in the linear regime. Mathematically, the output signal of the HPA is modeled as $x^{hpa}_i(t) =  Gx_i(t)$, where $G$ is the HPA gain. The rest of the signal modeling formulation will be presented separately for different transmit architectures in the following.

\subsubsection{Fully-Digital Architecture} 

Herein, $N_{rf} = N$, thus, the real transmit signal at the output of the $l$th element in the $i$th row of the UPA can be expressed as
\begin{align}
\quad x^{FD}_{i, l}(t) = \Re \biggl\{x^{hpa}_{u{(i,l)}}(t) \biggr\}  
= G \sum_{n=1}^{N_f} \Re \biggl\{\omega_{{u{(i,l)}},n}e^{j 2 \pi f_n t}\biggr\}, 
\end{align}
where ${u{(i,l)}} = (i - 1)N_h + l$. Thereby, the RF signal at the $m$th ER when exploiting the fully-digital architecture can be expressed as
\begin{align}\label{eq:FD_RCV}
    y^{FD}_m(t)  &= \sum_{i=1}^{N_v} \sum_{l=1}^{N_h} \sum_{n = 1}^{N_f} 
    \gamma_{i,l,m,n} x^{FD}_{i, l}(t) \nonumber \\ &=  G \sum_{i=1}^{N_v} \sum_{l=1}^{N_h} \sum_{n = 1}^{N_f}\Re \Bigl\{\gamma_{i,l,m,n} \omega_{{u{(i,l)}}, n} e^{j2\pi f_{n}t}\Bigr\}.
\end{align}    
Furthermore, by defining $\mathbf{w}_n = [\omega_{1, n}, \ldots, \omega_{N, n}]^T$ and $\boldsymbol{\gamma}_{m,n} = [\gamma_{m, n, 1, 1}, \gamma_{m, n, 1, 2},\ldots, \gamma_{m,n,N_v, N_h}]^T$, \eqref{eq:FD_RCV} can be reformulated as 
\begin{equation}\label{eq:FD_RCV_reform}
    y^{FD}_m(t)  =  G \sum_{n = 1}^{N_f}\Re \Bigl\{\boldsymbol{\gamma}_{m,n}^T \mathbf{w}_n e^{j2\pi f_{n}t}\Bigr\}.
\end{equation} 

\begin{figure}[t]
    \centering
    \includegraphics[width=0.6\columnwidth]{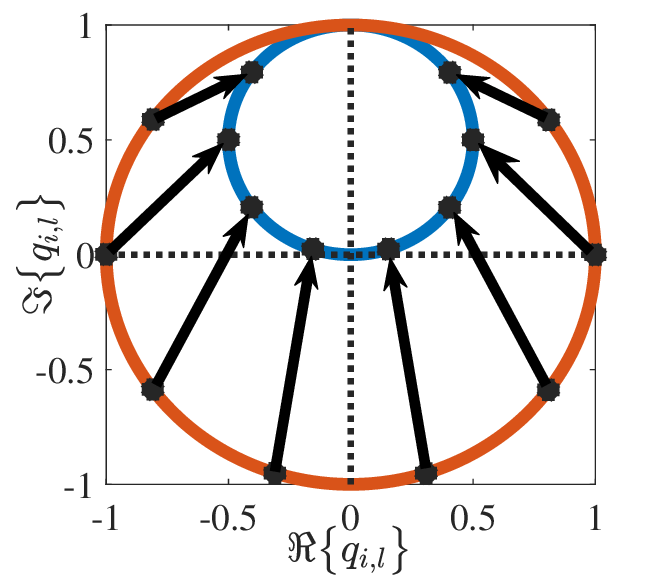}
    \caption{The Lorentzian constrained (the inner circle) and the ideal
weights (outer circle) in the complex plane. The arrows depict the mapping between
the weights.}
    \label{fig:phasemodel}
\end{figure}

\subsubsection{DMA-assisted Architecture}

In metasurface antennas, the phase and amplitude that can be configured in the radiating elements are correlated due to the Lorentzian resonance. Herein, we capture this correlation by \cite{DMAformulation}
\begin{equation}\label{eq:lorentzweight}
    q_{i, l} \in \mathcal{Q} = \Big\{{ (j+e^{j \phi_{i,l}})}/{2}\ \Big|\ \phi_{i,l} \in[0,2 \pi]\Big\}, \quad \forall i, l,
\end{equation}
where $\phi_{i,l}$ are the tunable phase of the $l$th metamaterial element in the $i$th waveguide. As shown in Fig. \ref{fig:phasemodel}, the ideal phase shifting exhibits a constant unit amplitude, i.e., no losses, while the amplitude of the Lorentzian weights depends on the configured phase.

Herein, microstrip lines are used as waveguides, similar to \cite{DMAuplink, near-field}. The propagation model of the signal within a microstrip is expressed as
% \begin{equation}\label{eq:propWG}
    $h_{i, l} = e^{-(l-1)d_l(\alpha_i+j\beta_i)}$,
% \end{equation}
where $d_l$ is the inter-element distance, $\alpha_i$ represents the waveguide attenuation coefficient, and $\beta_i$ is the propagation constant. The mathematical model of the DMA is represented in Fig.~\ref{fig:strcut}b. For the sake of simplicity and since the goal is to investigate the power consumption of DMAs for RF-WPT purposes, we consider in this work a simplified case with no mutual coupling between the elements of the metasurface. However, when investigating DMAs from an electromagnetic perspective, one should carefully consider the impact of mutual coupling \cite{direnzo_coupling}.

 Notice that the number of RF chains in the DMA is reduced to $N_{rf} = N_v$. Hence, the real transmit signal radiated from the $l$th element in the $i$th microstrip can be expressed as
\begin{align}\label{dmaQQ}
    x_{i,l}^{DMA}(t) &=  G \Re \left\{h_{i,l}q_{i,l}x_i^{hpa}(t)\right\} 
     \nonumber \\ &=G \sum_{n = 1}^{N_f} \Re \left\{h_{i,l}q_{i,l} {\omega_{i, n} e^{j2\pi f_n t}}\right\}.
\end{align}
Furthermore, the RF signal received at the $m$th ER in the DMA-assisted system is given by
\begin{align}\label{eq:dmarcv2}
    y^{DMA}_m(t) &= \sum_{i=1}^{N_v}\! \sum_{l=1}^{N_h}\! \sum_{n = 1}^{N_f} \gamma_{i,l,m,n} x_{i,l}^{DMA}(t) \nonumber \\
    &= G \sum_{i=1}^{N_v}\! \sum_{l=1}^{N_h}\! \sum_{n = 1}^{N_f} \!
        \Re \Bigl\{\gamma_{i,l,m,n} h_{i,l} q_{i,l} \omega_{i, n} e^{j2\pi f_{n}t} \Bigr\}.
\end{align}
Finally, we define 
\begin{align}
    \bar{\mathbf{w}}_n = [\underbrace{\omega_{1, n}, \ldots, \omega_{1, n}}_{N_h}, \ldots, \underbrace{\omega_{N_{v}, n}, \ldots, \omega_{N_{v}, n}}_{N_h}]^T \in \mathbb{C}^{N \times 1}, \nonumber
\end{align}
$\mathbf{q} = [q_{1, 1}, \ldots, q_{N_v, N_h}]^T$, $\mathbf{h} = [h_{1, 1}, \ldots, h_{N_v, N_h}]^T$, and reformulate \eqref{eq:dmarcv2} as 
\begin{equation}\label{eq:dmarcv2reform}
    y^{DMA}_m(t)  = G  \sum_{n = 1}^{N_f}
        \Re \Bigl\{(\boldsymbol{\gamma}_{m,n} \odot \mathbf{q} \odot \mathbf{h} )^T\bar{\mathbf{w}}_n e^{j2\pi f_{n}t} \Bigr\}.
\end{equation}

\subsection{Rectenna}

At the receiver side, the RF signal is converted to DC. This can be modeled by an antenna equivalent circuit and a single diode rectifier as illustrated in Fig. \ref{fig:rect}. The RF signal at the input of the antenna is denoted as $y_m(t)$ and has an average power of $\mathbb{E}\bigl\{{y_m(t)}^2\bigr\}$. Let us denote the input impedance of the rectifier and the impedance of the antenna equivalent circuit by $R_{in}$ and $R_{ant}$, respectively. Thus, assuming perfect matching $(R_{in}=R_{ant})$, the input voltage at the rectifier of the $m$th ER is given by $v_{{in},{m}}(t) = y_m(t)\sqrt{R_{ant}}$. Furthermore, the diode current can be formulated as
\begin{equation}
    i_d(t) = i_s\bigl( e^{\frac{v_d(t)}{\hat{n} v_t}} - 1 \bigr),
\end{equation}
where $i_s$ is the reverse bias saturation current, $\hat{n}$ is the ideality factor, $v_t$ is the thermal voltage, and $v_d(t) = v_{in}(t) - v_{o}(t)$ is the diode voltage. Moreover, $v_{{o},{m}}(t)$ is the output voltage of the $m$th rectifier, which can be approximated utilizing the Taylor expansion as \cite{jointWFandBFMIMOclreckx, largescale}
\begin{equation}\label{eq:rectifierbase}
    v_{{o},{m}} =  \sum_{i\ even, i\geq2}^{n_0} K_i \mathbb{E}_t\bigl\{y_m(t) ^i\bigr\},
\end{equation}
where $K_i = \frac{R_{ant}^{i/2}}{i!{(\eta_0 v_t)}^{i - 1}}$. Herein, we focus on the low-power regime, for which it was demonstrated in \cite{largescale, clreckxWFdesign} that truncating the Taylor expansion at $n_0=4$ is accurate enough. Therefore, (\ref{eq:rectifierbase}) can be written as
\begin{equation}\label{eq:vout}
    v_{{o},{m}} = K_2 \mathbb{E}_t\bigl\{y_m(t)^2\bigr\} + K_4 \mathbb{E}_t\bigl\{y_m(t)^4\bigr\}
\end{equation}
and the DC power at the $m$th ER is given by
\begin{equation}\label{eq:Pdc}
    P_{{dc},m} =\frac{{v^2_{{o},{m}}}}{R_L},
\end{equation}
where $R_L$ is the load impedance of the rectifier, while $y_m(t)$ is equal to $y_m^{DMA}(t)$ and $y_m^{FD}(t)$ in the DMA-assisted and fully-digital architectures, respectively. 
\begin{figure}[t]
    \centering
    \includegraphics[width=0.7\columnwidth]{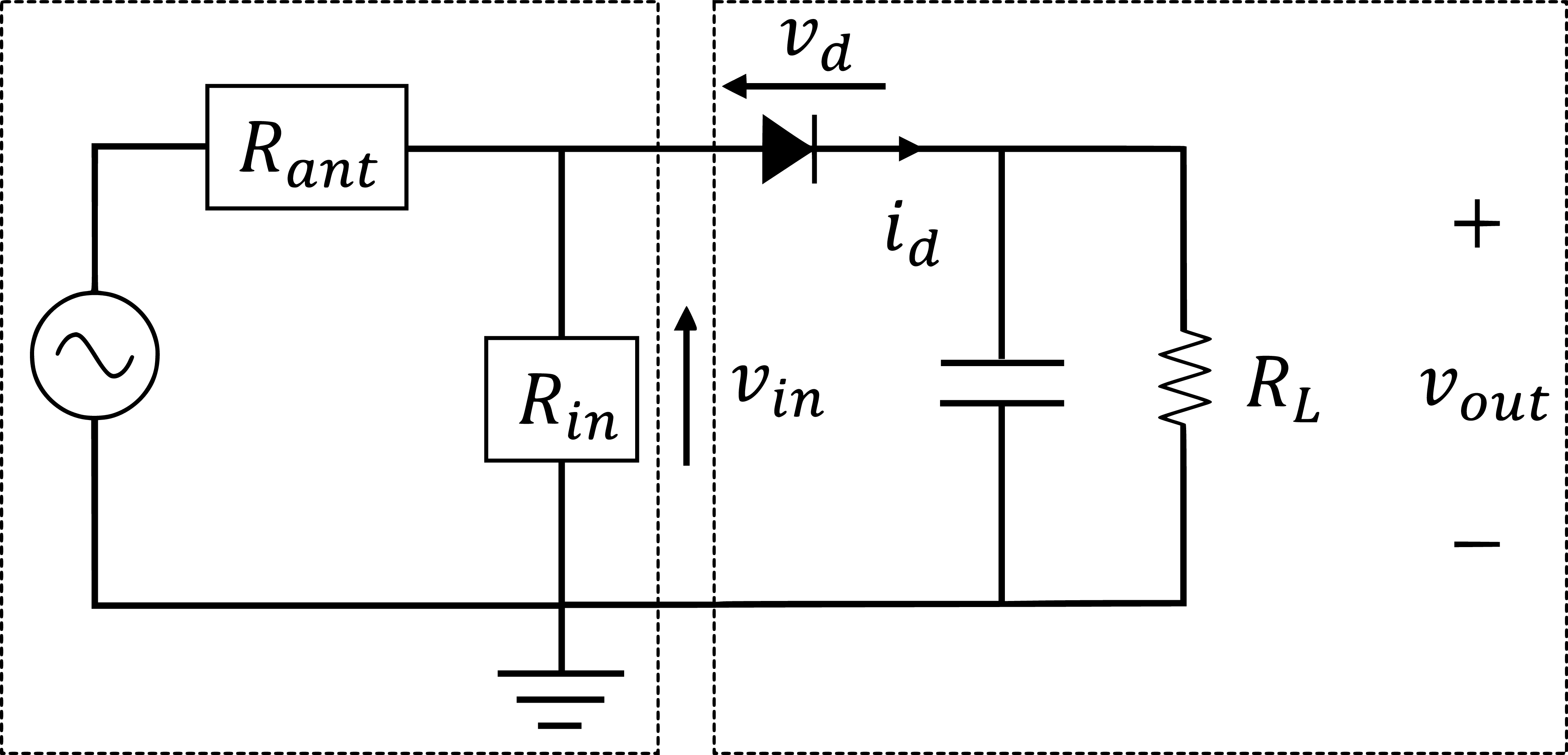}
    \caption{Antenna equivalent circuit (left) and a single diode rectifier (right) 
    \cite{jointWFandBFMIMOclreckx}.}
    \label{fig:rect}
\end{figure}

 \subsection{Power Consumption Model}\label{sec:powcons}
Lowering the system power consumption is desirable for increasing the end-to-end PTE and this greatly depends on the HPA design and operation. Let us denote the maximum efficiency and the maximum output power of a class-B HPA as $\bar{\eta}$ and $P_{max}$, respectively. Hereby, the efficiency of the $i$th HPA at time $t$ is expressed as $\eta_i(t) = {\bar{\eta}}\sqrt{{P_{out, i}(t)}/{P_{max}}}$ \cite{dohertyConsum}, where $P_{out, i}(t)$ is the output power of the HPA. Then, the corresponding power consumption is
\begin{equation}\label{eq:hpacons}
    P_{hpa,i}(t) = \frac{P_{out, i}(t)}{\eta_i(t)} = \frac{1}{\bar{\eta}}\sqrt{{P_{max}{P_{out, i}(t)}}}.
\end{equation}
Notice that each HPA is in charge of supplying the radiating power by all the antenna elements, which are fed by its corresponding RF chain. Thus, $P_{out, i}(t)$ is $\sum_{l = 1}^{N_h} {x^{DMA}_{i,l}(t)}^2$ in DMA-assisted architecture, and  ${x^{FD}_{i, l}(t)}^2$ denotes the output power of the RF chain connected to the $l$th element in the $i$th row.

There are also other power consumption sources in the WPT system. For instance, digital baseband power consumption, which is required to perform the digital beamforming, and the RF chain circuit power consumption, including the mixer, local oscillator, and filter. However, the power consumption of these sources is usually considered fixed and is negligible compared to the HPA. Thus, without loss of generality, the total power consumption of the system is given by
\begin{equation}\label{eq:power_consumption}
    P_c =   \sum_{i = 1}^{N_{rf}} \mathbb{E}_t\bigl\{P_{{hpa},i}(t)\bigr\} + P_{in},
\end{equation}
where 
%\begin{equation}
   $P_{in} = \sum_{i = 1}^{N_{rf}} \sum_{n = 1}^{N_f} |\omega_{i,n}|^2$
%\end{equation} 
is the total input power. 

\section{Joint Waveform \& Beamforming Optimization}\label{Optimization}

This section formalizes the optimization problem and describes the utilized approach when employing the aforementioned transmit architectures. 

\subsection{Problem Formulation}

The goal is to obtain a minimum-power waveform and beamforming design to meet the EH requirements of the ERs. For this, we assume that the locations of the receivers are fixed and known at the transmitter, which leads to the perfect CSI assumption. Thus, by utilizing \eqref{eq:Pdc} and substituting \eqref{eq:hpacons} into \eqref{eq:power_consumption}, the optimization problem can be formulated as
\begin{subequations}\label{probmain}
\begin{align}
\label{probmaina} \minimize_{\mathcal{V}} \quad &  \frac{\sqrt{P_{max}}}{\bar{\eta}} \sum_{i = 1}^{N_{rf}} \mathbb{E}_t\bigl\{\sqrt{P_{out, i} (t)}\bigr\} + \sum_{i = 1}^{N_{rf}} \sum_{n = 1}^{N_f} |\omega_{i,n}|^2\\
\textrm{subject to} \label{probFDb} \quad &v_{o, m}^2  \geq R_L \bar{P}_{m} , \quad  m = 1, \ldots, M,
\end{align}
\end{subequations}
 where $\mathcal{V}$ is the set of optimization variables, which is equal to $\{\omega_{i,n}, q_{i,l} \}_{\forall i,l,n}$ and $\{\omega_{i,n}\}_{\forall i,n}$ for DMA-assisted and fully-digital architectures, respectively. Problem \eqref{probmain} deals with extensive non-linearity since the objective and constraints are highly non-linear and non-convex functions due to the signal model, rectenna non-linearity, and the coupling between the digital and analog beamforming variables in DMA. Notice that this problem formulation is equivalent to maximizing the PTE. Specifically, the solution will force the DC harvested power of the ERs to be equal to the EH requirement ($\bar{P}_m$). Thus, the numerator of \eqref{eq:wptbase} is constant, and decreasing the power consumption leads to maximizing the PTE.

\subsection{Optimization Framework for DMA-assisted Architecture}

One of the challenges making the problem intractable is the coupling between the optimization variables. To cope with this, we propose using alternating optimization by first optimizing the digital precoders when considering fixed $\{q_{i,l}\}_{\forall i, l}$, followed by optimizing the metamaterial elements' frequency response for fixed $\{\omega_{i,n}\}_{\forall i, n}$.

\subsubsection{Optimization with fixed $\{q_{i,l}\}_{\forall i, l}$}

Let us proceed by rewriting the optimization problem as 
\begin{subequations}\label{probdma}
\begin{align}
\label{probdmaa}\minimize_{\{ \omega_{i,n}\}_{\forall i,n}} \quad & \hspace{-2mm}  \frac{\sqrt{P_{max}}}{\bar{\eta}} \sum_{i = 1}^{N_{v}} \mathbb{E}_t\biggl\{\sqrt{\sum_{l = 1}^{N_h} x^{DMA}_{i,l}(t)^2}\biggr\} \!+\! \sum_{i = 1}^{N_{v}} \sum_{n = 1}^{N_f} |\omega_{i,n}|^2 \\
\textrm{subject to} \label{probdmab} \quad & \hspace{-2mm} {K_2 \mathbb{E}_t\bigl\{y_m(t)^2\bigr\} + K_4 \mathbb{E}_t\bigl\{y_m(t)^4\bigr\}} \geq \sqrt{{R_L}\bar{P}_{m}} , \forall m.
\end{align}
\end{subequations}
The problem is still highly non-linear and non-convex. Interestingly, it can be easily verified that $v_{o, m}$ is convex and increasing with respect to ${y_m(t)}^2$, and ${y_m(t)}$ is affine with respect to $\bar{\mathbf{w}}$. Therefore, $v_{o, m}$ is a convex function with respect to  $\bar{\mathbf{w}}$ given a fixed $q_{i,l}, \forall i,l$ \cite{SISOAllClreckx}. Thus, although \eqref{probdmab} is a non-convex constraint, its left-hand side consists of a convex function. These properties motivate us to adapt the SCA method \cite{boyd2004convex} to optimize the problem iteratively. Specifically, the constraints can be approximated by their first-order Taylor expansion.

Let us proceed by rewriting $v_{o, m}$ considering fixed $q_{i,l}$. For this, we define $\mathbf{a}_{m,n} = \boldsymbol{\gamma}_{m,n} \odot \mathbf{q} \odot \mathbf{h}$ and leverage the fact that the average power of a signal is equal to the power of its spectrum. Thus, we have \cite{clreckxWFdesign}
\begin{align}
\label{eq:Erewrit1}&\mathbb{E}_t\bigl\{y_m(t)^2\bigr\} = g_{m,1}(\{\bar{\mathbf{w}}_n\}_{\forall n}) = \frac{G^2}{2} \sum_{n}  |\mathbf{a}^T_{m,n} \bar{\mathbf{w}}_n|^2, 
\\
    \label{eq:Erewrit2}&\mathbb{E}_t\bigl\{y_m(t)^4\bigr\} = g_{m,2}(\{\bar{\mathbf{w}}_n\}_{\forall n}) = \frac{3G^4}{8} \sum_{\substack{n_0, n_1, n_2, n_3 \\ n_0 + n_1 = n_2 + n_3}} ... \nonumber \\ &(\mathbf{a}^T_{m,n_0} \bar{\mathbf{w}}_{n_0})(\mathbf{a}^T_{m,n_1} \bar{\mathbf{w}}_{n_1})(\mathbf{a}^T_{m,n_2} \bar{\mathbf{w}}_{n_2})^\star(\mathbf{a}^T_{m,n_3} \bar{\mathbf{w}}_{n_3})^\star.
\end{align}
By leveraging \eqref{eq:Erewrit1} and \eqref{eq:Erewrit2}, we can write
\begin{multline}\label{eq:powerineq}
    {{v}_{o, m}} \geq K_2 g_{m,1}(\{\bar{\mathbf{w}}_n^{(0)}\}_{\forall n}) + {K_4} g_{m,2}(\{\bar{\mathbf{w}}_n^{(0)}\}_{\forall n}) \\ + \sum_{n} \tilde{g}_{m,n}(\bar{\mathbf{w}}_n^{(0)})\bigl(\bar{\mathbf{w}}_n - \bar{\mathbf{w}}_n^{(0)}\bigr),
\end{multline}
where
\begin{multline}
    \tilde{g}_{m,n}(\bar{\mathbf{w}}_n^{(0)}) = G^2 K_2 (\mathbf{a}^T_{m,n}\bar{\mathbf{w}}_n^{(0)})\mathbf{a}^T_{m,n} + \\ \frac{3 K_4 G^4}{8}\biggl[4{|\mathbf{a}^T_{m,n}|}^4{|\bar{\mathbf{w}}_n^{(0)}|}^2{{\bar{\mathbf{w}}{_n^{(0)}}}}^T + \\
    8 \sum_{n1} {|\mathbf{a}^T_{m,n}|}^2{|\mathbf{a}^T_{m,n_1}|}^2{|\bar{\mathbf{w}}_{n_1}^{(0)}|}^2{{\bar{\mathbf{w}}{_n^{(0)}}}}^T + \\
    \sum_{\substack{n_2, n_3 \\ n_2 + n_3 = 2n \\ n_2 \neq n_3}} 2{(\mathbf{a}^T_{m,n_2} \bar{\mathbf{w}}_{n_2}^{(0)})}^\star
    {(\mathbf{a}^T_{m,n_3} \bar{\mathbf{w}}_{n_3}^{(0)})}^\star
    {(\mathbf{a}^T_{m,n}\bar{\mathbf{w}}_n^{(0)})}\mathbf{a}^T_{m,n} + \\
    2(\mathbf{a}^T_{m,n_2}\bar{\mathbf{w}}_{n_2}^{(0)})(\mathbf{a}^T_{m,n_3}\bar{\mathbf{w}}_{n_3}^{(0)}){{(\mathbf{a}^T_{m,n}\bar{\mathbf{w}}_n^{(0)})}}^\star {\mathbf{a}^H_{m,n}} + \\
    \sum_{\substack{n_1, n_2, n_3 \\ -n_1 + n_2 + n_3 = n \\ n \neq n_1 \neq n_2 \neq n_3}} 2(\mathbf{a}^T_{m,n_1}\bar{\mathbf{w}}_{n_1}^{(0)})
    {(\mathbf{a}^T_{m,n_2}\bar{\mathbf{w}}_{n_2}^{(0)})}^\star 
    {(\mathbf{a}^T_{m,n_3}\bar{\mathbf{w}}_{n_3}^{(0)})}^\star \mathbf{a}^T_{m,n} + \\
    2 (\mathbf{a}^T_{m,n_1}\bar{\mathbf{w}}_{n_1}^{(0)})
    (\mathbf{a}^T_{m,n_2}\bar{\mathbf{w}}_{n_2}^{(0)}) 
    {(\mathbf{a}^T_{m,n_3}\bar{\mathbf{w}}_{n_3}^{(0)})}^\star  {\mathbf{a}^H_{m,n}}\biggr]
\end{multline}
is the first-order Taylor coefficient of ${{v}_{o, m}}$ at point $\{\bar{\mathbf{w}}_n^{(0)}\}_{\forall n}$. 

\begin{theorem}\label{theorem:1}
The objective function in \eqref{probdmaa} can be upper bounded by using
\begin{equation}\label{eq:reformobj}
    \sum_{i = 1}^{N_{rf}} \mathbb{E}_t\biggl\{\sqrt{\sum_{l = 1}^{N_h} x^{DMA}_{i,l}(t)^2}\biggr\}  \leq   \sum_{i = 1}^{N_{rf}} \sqrt{\sum_{l = 1}^{N_h} \sum_{n = 1}^{N_f}\frac{G^2}{2}|\omega_{i,n} q_{i,l}h_{i,l}|^2},
\end{equation}
where the right-hand side of the inequality is convex w.r.t. $\omega_{i,n}$.
\end{theorem}
\begin{proof}
    The proof is provided in Appendix~\ref{appen1}.
\end{proof}

Now, we can reformulate the problem at point $\{ \bar{\mathbf{w}}_n^{(0)}, S_i^{(0)}, {v}_{o, m}^{(0)}\}_{\forall i, n, m}$ as
\begin{subequations}\label{probdmaSCA}
\begin{align}
\label{probdmaSCAa} \minimize_{ \{\bar{\mathbf{w}}_n\}_{\forall n}} \quad &  \Upsilon \!=\!  \frac{\sqrt{P_{max}}}{\bar{\eta}}\sum_{i = 1}^{N_{rf}} \sqrt{\sum_{l = 1}^{N_h} \sum_{n = 1}^{N_f}\frac{G^2}{2}|\omega_{i,n} q_{i,l}h_{i,l}|^2} + P_{in} \\
\textrm{subject to} \label{probdmaSCAc} \quad & \sqrt{{R_L}\bar{P}_{m}}  \leq K_2g_{m,1}(\{\bar{\mathbf{w}}_n^{(0)}\}_{\forall n}) + \nonumber \\  &K_4g_{m,2}(\{\bar{\mathbf{w}}_n^{(0)}\}_{\forall n}) + \nonumber \\ & \sum_{n} \tilde{g}_{m,n}(\bar{\mathbf{w}}_n^{(0)})\bigl(\bar{\mathbf{w}}_n - \bar{\mathbf{w}}_n^{(0)}\bigr), \quad \forall m,\\
\label{probdmaSCAe} \quad & \bar{\mathbf{w}}_n [(i - 1)N_h + l] = \bar{\mathbf{w}}_n [i], \forall i, l = 1, \ldots, N_h.
\end{align}
\end{subequations}
Notice that by utilizing \eqref{eq:reformobj} and the fact that \eqref{probdmaa} consists of two positive terms, one has that \eqref{probdmaSCAa} serves as an upper bound for \eqref{probdmaa}. Moreover, the inequality in \eqref{eq:powerineq} ensures that the solution to this problem is a feasible solution of \eqref{probdma} at each point. Interestingly, the problem has become convex and can be solved at a given point by standard convex optimization tools, e.g., CVX \cite{cvxref}.\footnote{For the sake of easy notation and facilitating the reader's understanding, we have kept the problem in vector form and introduced the constraint \eqref{probdmaSCAe} into the problem. However, the problem can be easily converted to scalar form, which removes the mentioned constraint.} Moreover, the solution can be iteratively updated using the SCA algorithm \cite{boyd2004convex}.

\subsubsection{Optimization with fixed $\{\omega_{i,n}\}_{\forall i, n}$}

Herein, the non-convex Lorentzian constraint of the metamaterials' frequency response makes the problem extremely difficult to solve. To tackle this, we propose decoupling the problem to first maximize the minimum harvested power when optimizing $q_{i,l}$. This allows us to leverage the beamforming capability of the metamaterial elements and provide degrees of freedom to further reduce the power consumption when optimizing $\omega_{i,n}$ \cite{MyEBDMA}. Hereby, the optimization problem with fixed $\{\omega_{i,n}\}_{\forall i, n}$ can be reformulated as 
\begin{subequations}\label{probdmaQcon}
\begin{align}
\label{probdmaQcona} \maximize_{\{ q_{i,l}\}_{\forall i,l}} \quad & \min_{m} \frac{v_{o, m}^2}{R_L} \\
\textrm{subject to} \label{probdmaQg}  \quad & q_{i,l} \in \mathcal{Q}, \forall i,l,
\end{align}
\end{subequations}
where \eqref{probdmaQg} is the non-convex Lorentzian constraint. Next, we cope with the complexity caused by the metamaterial elements. 
\begin{theorem}\label{theorem:2}
Problem \eqref{probdmaQcon} is equivalent to
\begin{subequations}\label{probdmaQQcon}
\begin{align}
\label{probdmaQQcona} \maximize_{\{ q_{i,l}\}_{\forall i,l}} \quad &  R \\
\textrm{subject to} \label{probdmaQQQg}  \quad & \Re\bigl\{q_{i,l}\bigr\}^2 + (\Im\bigl\{q_{i,l}\bigr\} - 0.5)^2 \leq 0.25, \forall i,l, \\
\label{probdmaQQconf} \quad & R \leq  K_2 {E}\bigl\{y_m(t)^4\bigr\} + K_4{E}\bigl\{y_m(t)^4\bigr\}, \quad \forall m.
\end{align}
\end{subequations}
\end{theorem}
\begin{proof}
    The proof is provided in Appendix~\ref{appen2}.
\end{proof}

Notice that \eqref{probdmaQQQg} keeps the frequency response of the metamaterials within the Lorentzian circle, while \eqref{probdmaQQconf} ensures that $R$ is below the minimum output voltage of the ERs. Problem \eqref{probdmaQQcon} is still difficult to solve due to the non-convex constraint \eqref{probdmaQQconf}. To cope with this, we define $\hat{\mathbf{a}}_{m,n} = \boldsymbol{\gamma}_{m,n} \odot \bar{\mathbf{w}}_n \odot \mathbf{h}$ and write
\begin{align}
\label{eq:Erewrit1q}&\mathbb{E}_t\bigl\{y_m(t)^2\bigr\} = {e}_{m,1}({\mathbf{q}}) = \frac{G^2}{2} \sum_{n}  |\hat{\mathbf{a}}^T_{m,n} {\mathbf{q}}|^2, 
\\
    \label{eq:Erewrit2q}&\mathbb{E}_t\bigl\{y_m(t)^4\bigr\} = {e}_{m,2}(\mathbf{q}) = \frac{3G^4}{8} \sum_{\substack{n_0, n_1, n_2, n_3 \\ n_0 + n_1 = n_2 + n_3}} ... \nonumber \\ &(\hat{\mathbf{a}}^T_{m,n_0} \mathbf{q})(\hat{\mathbf{a}}^T_{m,n_1} \mathbf{q})(\hat{\mathbf{a}}^T_{m,n_2} \mathbf{q})^\star(\hat{\mathbf{a}}^T_{m,n_3} \mathbf{q})^\star.
\end{align}
Similar to the case of digital precoders, it can be observed that ${v}_{o, m}$ is convex with respect to $\mathbf{q}$, thus, we can write 
\begin{multline}\label{eq:powerineqQ}
    \hspace{-3mm}{{v}_{o, m}} \geq K_2 {e}_{m,1}({\mathbf{q}}^{(0)}) + K_4 {e}_{m,2}({\mathbf{q}}^{(0)}) +  \tilde{e}_{m}(\mathbf{q}^{(0)})\bigl({\mathbf{q}} - {\mathbf{q}}^{(0)}\bigr),
\end{multline}
where
\begin{multline}
    \tilde{e}_{m}(\mathbf{q}^{(0)}) = G^2K_2\sum_n (\hat{\mathbf{a}}^T_{m,n}{\mathbf{q}}^{(0)})\hat{\mathbf{a}}^T_{m,n} + \\ \frac{3K_4G^4}{8} \sum_{\substack{n_0, n_1, n_2, n_3 \\ n_0 + n_1 = n_2 + n_3}}\biggl[(\hat{\mathbf{a}}^T_{m,n_1} \mathbf{q})(\hat{\mathbf{a}}^T_{m,n_2} \mathbf{q})^\star(\hat{\mathbf{a}}^T_{m,n_3} \mathbf{q})^\star \hat{\mathbf{a}}^T_{m,n_0} + \\(\hat{\mathbf{a}}^T_{m,n_0} \mathbf{q})(\hat{\mathbf{a}}^T_{m,n_2} \mathbf{q})^\star(\hat{\mathbf{a}}^T_{m,n_3} \mathbf{q})^\star \hat{\mathbf{a}}^T_{m,n_1} + \\ 
    (\hat{\mathbf{a}}^T_{m,n_0} \mathbf{q})(\hat{\mathbf{a}}^T_{m,n_1} \mathbf{q})(\hat{\mathbf{a}}^T_{m,n_3} \mathbf{q})^\star \hat{\mathbf{a}}_{m,n_2}^H + \\
    (\hat{\mathbf{a}}^T_{m,n_0} \mathbf{q})(\hat{\mathbf{a}}^T_{m,n_1} \mathbf{q})(\hat{\mathbf{a}}^T_{m,n_2} \mathbf{q})^\star \hat{\mathbf{a}}_{m,n_3}^H\biggr].
\end{multline}
Hereby, \eqref{probdmaQQcon} can be reformulated at point $\mathbf{q}^{(0)}$ as
\begin{subequations}\label{RRprobdmaQQcon}
\begin{align}
\label{RRprobdmaQcona} \maximize_{\mathbf{q}, R} \quad &  R \\
\textrm{subject to} 
\label{RRprobdmaQQg}  \quad &\hspace{-2mm} R \leq K_2\bar{g}_{m,1}({\mathbf{q}}^{(0)})  \nonumber \\ &\hspace{-2mm}+ K_4 \bar{g}_{m,2}({\mathbf{q}}^{(0)}) + \tilde{e}_{m}(\mathbf{q}^{(0)})\bigl({\mathbf{q}} - {\mathbf{q}}^{(0)}\bigr), \\
\label{RRprobdmaQQh}  \quad &\hspace{-2mm} \Re\bigl\{\mathbf{q}[(i-1)N_h + l]\bigr\}^2 + \nonumber \\ &\hspace{-2mm}(\Im\bigl\{\mathbf{q}[(i-1)N_h + l]\bigr\} - 0.5)^2 \leq 0.25, \forall i,l, \eqref{probdmaQQconf}
\end{align}
\end{subequations}
which is convex in the neighborhood of $\mathbf{q}^{(0)}$.

\subsubsection{Alternating Optimization Algorithm}

We transformed the original problem into a convex form for both fixed $\mathbf{q}$ and $\bar{\mathbf{w}}$ and in the neighborhood of $\bar{\mathbf{w}}^{(0)}$ and $\mathbf{q}^{(0)}$ in \eqref{probdmaSCA} and \eqref{RRprobdmaQQcon}, respectively. However, there is still an important challenge, i.e., the initialization of the variables. Specifically, when using iterative algorithms relying on convex approximation, e.g., SCA, the starting point must be feasible and the performance is influenced by it. One can get a feasible point by setting an extremely large amplitude for the digital weights, but this may lead to poor performance. To cope with this, we propose a low-complexity initialization method, which will be explained in the next section. This prevents the initial consumed power from becoming extremely large, helping the SCA algorithm to start from a relatively good initial point.

Algorithm \ref{alg:DMA} illustrates the proposed alternating SCA-based approach for waveform and beamforming optimization. First, the digital precoders and the frequency responses of the metamaterial elements are initialized. After that, digital precoding weights and the frequency response of the metamaterials are optimized in an alternating fashion in lines \ref{alg1:line:alter_start}-\ref{alg1:line:alter_end}. First, SCA is used to iteratively find a suboptimal solution for fixed digital precoders. Specifically, $\mathbf{q}$ is updated in each iteration until convergence in lines \ref{alg1:line:QSCA_start}-\ref{alg1:line:QSCA_end}. Then, the obtained $\mathbf{q}$ is used to run the SCA algorithm for finding suboptimal digital precoders $\{\bar{\mathbf{w}}_n\}_{\forall n}$ through lines \ref{alg1:line:digitSCA_start}-\ref{alg1:line:digitSCA_end}. These two SCA-based optimizations are repeated until the alternating optimization converges to a suboptimal solution.

\begin{algorithm}[t]
	\caption{Alternating SCA-based waveform and beamforming design for DMA-assisted WPT (ASCA-DMA).} \label{alg:DMA}
	\begin{algorithmic}[1]
            \State \textbf{Input:} $\{\gamma_{i,l,m,n}\}_{\forall i, l, m, n}$, $\upsilon$
            \quad \textbf{Output:} $\{\bar{\mathbf{w}}_{n}^{(0)}\}_{\forall n}$, $\mathbf{q}^{(0)}$
            \State \textbf{Initialize:} $\mathbf{q}^{(0)}$ and $\bar{\mathbf{w}}^{(0)}_{n}, \forall n$, $P_c^\star =0$\label{alg1:line:init_start}
            \Repeat\label{alg1:line:alter_start}
                \State \hspace{-2mm} $\Upsilon^\star =0$, $\xi^\star =\infty$, $P_c \leftarrow P_c^\star$
                \Repeat\label{alg1:line:QSCA_start}
                    \State \hspace{-2mm}  $\xi \leftarrow \xi^\star$ 
                    \State \hspace{-2mm} Solve \eqref{RRprobdmaQQcon} to obtain $\mathbf{q}$ and  $\mathbf{q}^{(0)} \leftarrow \mathbf{q}$
                    \State \hspace{-2mm} $\xi^\star \leftarrow$ the objective value in \eqref{RRprobdmaQcona}
                \Until{$|1 - {\xi^\star}/{\xi}|\leq \upsilon$}\label{alg1:line:QSCA_end}
                \Repeat\label{alg1:line:digitSCA_start}
                    \State \hspace{-2mm}$\Upsilon \leftarrow \Upsilon^\star$
                    \State \hspace{-2mm} solve \eqref{probdmaSCA} to obtain $\bar{\mathbf{w}}_{n}$ and $\bar{\mathbf{w}}^{(0)}_{n} \leftarrow \bar{\mathbf{w}}_{n}, \forall n$
                    \State \hspace{-2mm} Compute $\Upsilon^\star$ using \eqref{probdmaSCAa}
                \Until{$|1 - {\Upsilon^\star}/{\Upsilon}|\leq \upsilon$}\label{alg1:line:digitSCA_end}
                \State \hspace{-2mm} $P_c^\star \leftarrow \Upsilon^\star$ 
            \Until{$|1 - P_c^\star/ P_c|\leq \upsilon$}\label{alg1:line:alter_end}
\end{algorithmic} 
\end{algorithm}

\subsection{Optimization Framework for Fully-Digital Architecture}

Note that the proposed framework for the DMA-assisted system can also be used for fully-digital architecture with some slight modifications.\footnote{In fact, the framework can also be straightforwardly adapted for a traditional hybrid architecture with a fully connected network of phase shifters.} Notably, alternating optimization is not needed since the variables are just digital precoders, and adapting SCA to find the suboptimal precoders is sufficient. Algorithm~\ref{alg:FD} illustrates the SCA-based optimization for fully-digital architecture. Notice that problem \eqref{probdmaSCA} can be easily modified to match the fully-digital transmitter. Thus, the expressions are not rewritten to avoid repetition.

\begin{algorithm}[t]
	\caption{SCA-based waveform and beamforming design for fully-digital WPT (SCA-FD).} \label{alg:FD}
	\begin{algorithmic}[1]
            \State \textbf{Input:} $\{\gamma_{i,l,m,n}\}_{\forall i, l, m, n}$, $\upsilon$
            \quad \textbf{Output:} $\{\mathbf{w}^{(0)}\}_{\forall n}$
            \State \textbf{Initialize:} ${\mathbf{w}}^{(0)}_{n}, \forall n$, $P_c^\star =0$\label{alg1:line:init_start}
            \Repeat
                \State \hspace{-2mm} $\Upsilon \leftarrow \Upsilon^\star$
                , solve \eqref{probdmaSCA} to obtain $\{{\mathbf{w}}_{n}\}_{\forall n}$
                \State \hspace{-2mm} ${\mathbf{w}}^{(0)}_{n} \leftarrow {\mathbf{w}}^{(0)}_{n}, \forall n$
                \State \hspace{-2mm} Compute $\Upsilon^\star$ using \eqref{probdmaSCAa}
            \Until{$|1 - \Upsilon^\star/ \Upsilon|\leq \upsilon$}
\end{algorithmic} 
\end{algorithm}

\subsection{Complexity}

The proposed ASCA-DMA algorithm consists of a low-complexity initialization followed by alternating optimization, while SCA is used to optimize each set of variables. Each iteration of SCA attempts to solve a quadratic program \cite{boyd2004convex} (\eqref{probdmaSCA} or \eqref{RRprobdmaQQcon}), whose complexity scales with a polynomial function of the problem size. In general, the degree of the polynomial mainly depends on the type of the solver. Let us consider a simple solver based on the Newton method with $\mathcal{O}(n^3)$ complexity \cite{boyd2004convex}, where $n$ is the problem size. Assume $U_1$ and $U_2$ are the number of required iterations (in the worst-case) for convergence of digital weights and metamaterial elements' weights, respectively. Furthermore, $U_3$ is considered the number of iterations required for convergence in alternating optimization. Hereby, the total complexity of the ASCA-DMA algorithm is $\mathcal{O}(U_1U_2U_3n^3)$, where $n$ scales with $M$, $N_v$, $N_h$, and $N_f$. There is also some additional complexity introduced by the initialization algorithm, which is negligible since the initialization procedure is low-complexity. Moreover, the SCA-FD has only a single SCA stage with a complexity $\mathcal{O}(U_4n^3)$, where $U_4$ is the required number of iterations for convergence of SCA in the worst-case.

\section{Initialization Algorithm}\label{sec:init}

Algorithm~\ref{alg:initial} illustrates the proposed initialization algorithm. Since the initialization algorithm has to be adaptable for multi-user scenarios, we start by proposing a method to allocate the output signal of the RF chains to the different ERs\footnote{The allocation is only for the initialization process, and there is no limitation in this regard in the optimization procedure.}. For this, we utilize the channel characteristics by naming $n_m^\star = \argmax_{n} |\boldsymbol{\gamma}_{m,n}|$
% \begin{equation}\label{eq:n_maxderiv}
%     n_m^\star = \argmax_{n} |\boldsymbol{\gamma}_{m,n}|,\quad \forall m
% \end{equation}
as the strongest sub-carrier channel between ER $m$ and the transmitter. Then, a coefficient $z_m$ is assigned to ER $m$ based on the gain introduced by its strongest channel, expressed as
\begin{equation}\label{eq:ratioalloc}
    z_m = 1 - \frac{|\boldsymbol{\gamma}_{m,n_m^\star}|}{\sum_{\bar{m} = 1}^M|\boldsymbol{\gamma}_{\bar{m},{n_{\bar{m}}^\star}}|}.
\end{equation}
More precisely, the RF chains are dedicated to the ERs based on this ratio such that the ERs with lower channel gains are served by more signals and vice versa. The allocation procedure is illustrated in lines \ref{alginit:startaloc}-\ref{alginit:endaloc} in Algorithm~\ref{alg:initial}. First, an RF chain is allocated to each ER, then, the rest of the RF chains are divided among ERs based on their $z_m$.

Let us denote $\mathcal{R}_m$ as the set of RF chains dedicated to ER $m$. Then, we initialize $q_{i,l}, i \in \mathcal{R}_m$ to compensate for the phase shift introduced by both $h_{i,l}$ and $\gamma_{i,l,m,n_m^\star}$. Specifically, we need to define $q_{i,l}, i\in \mathcal{R}_m$ such that 
\begin{equation}\label{eq:qinit}
    \phi^\star_{i,l} = \argmin_{\phi_{i,l}} \biggl\langle \bigl(\frac{j + e^{j\phi_{i,l}}}{2}\bigr)h_{i,l}\gamma_{i,l,m,n_m^\star}\biggr\rangle, \forall i,l,
\end{equation}
where $i \in \mathcal{R}_m$ and $q_{i,l}$ can be obtained accordingly. Notice that \eqref{eq:qinit} can be easily solved using a one-dimensional search with negligible complexity.

The next step is to initialize the amplitude and phase of the digital precoders. For this, let us proceed by defining the received RF power at the $m$th ER as
\begin{equation}
    P_{rf, m} =\frac{G^2}{2}  |\mathbf{a}_{m,n} \bar{\mathbf{w}}_n|^2 \!=\! \frac{G^2}{2} \sum_{n = 1}^{N_f}  \biggl|\sum_{i=1}^{N_v} \sum_{l=1}^{N_h} \gamma_{i,l,m,n} q_{i,l} h_{i,l} \omega_{i, n}\biggr|^2.
\end{equation}
Moreover, the output DC power of the rectifier is an increasing function of the input RF power when operating in the low-power regime below the breakdown region of the rectifier circuit's diode \cite{clerckx2018beneficial}. Motivated by this, we aim to increase the available RF power at each ER during the initialization process using the phases of the digital precoders with low complexity. One way to increase the available RF power for each ER is to facilitate the coherent reception of the signals at the receiver. This can be done by reducing the amount of phase shift introduced in the signal.  For this, we assume the initial digital weights to be $\omega_{i, n} = \bar{\omega}_{i, n}e^{j\hat{\omega}_{i, n}}$, while the dedicated signals to each ER have the same amplitude $\bar{\omega}_{i, n} = w_m, \forall n, i \in \mathcal{R}_m$. Moreover, the ideal phase initialization for maximizing the received RF power is obtained by solving
\begin{equation}\label{opt:RFpow}
    \argmax_{\hat{\omega}_{i, n} \in [0, 2\pi], \forall i \in \mathcal{R}_m, n} \sum_{n = 1}^{N_f}  \biggl|\sum_{i=1}^{N_v} \sum_{l=1}^{N_h} \gamma_{i,l,m,n} q_{i,l} h_{i,l}  \bar{\omega}_{i, n} e^{j\hat{\omega}_{i, n}}\biggr|^2.
\end{equation}
Meanwhile, solving this problem is not straightforward and introduces much additional complexity to our framework. Notably, the problem can be decoupled and solved individually for each sub-carrier without any change in the optimal solution. Still, there is a coupling between the digital weights of different RF chains given the same sub-carrier. For this, we further reduce the complexity by formulating the problem as
\begin{equation}\label{opt:low-comp_phase}
    \argmin_{\hat{\omega}_{i, n} \in [0, 2\pi]} \biggl|\bigg\langle\sum_{l=1}^{N_h} \gamma_{i,l,m,n} q_{i,l} h_{i,l} e^{j\hat{\omega}_{i, n}}\bigg\rangle\biggr|, \forall n,i \in \mathcal{R}_m.
\end{equation}
Although the reformulated problem may not have the same solution as \eqref{opt:RFpow}, it attempts to reduce the total amount of the phase shift of the received signal. Therefore, solving \eqref{opt:low-comp_phase} can lead to a suboptimal solution to \eqref{opt:RFpow} with much lower complexity, and by using a one-dimensional search. Note that utilizing such an approach is relevant since the goal is to have a reasonable initialization for the variables, which leads to feeding a feasible initial point to the optimization algorithm. The initialization procedure is illustrated through lines \ref{alginit:startpow}-\ref{alginit:endpow} in Algorithm~\ref{alg:initial}. For each ER, the frequency responses of the connected metamaterials to its dedicated RF chains are initialized. Then, the phases of the corresponding digital precoders are obtained. Finally, the amplitudes are iteratively increased until the EH requirement is met and a feasible solution is found. In the fully-digital architecture, the initialization algorithm follows the same procedure with one difference, i.e., each antenna element has a dedicated RF chain, and thus, a dedicated signal.

\begin{algorithm}[t]
	\caption{Initialization of digital precoders and the frequency response of the metamaterial elements.} \label{alg:initial}
	\begin{algorithmic}[1]
            \State \textbf{Input:} $\{\gamma_{i,l,m,n}\}_{\forall i, l, m, n}$, $\tau_s$, $\varsigma$
             \textbf{Output:} $\{\omega_{i,n}, q_{i,l}\}_{\forall i,l,n}$
            \State \textbf{Initialize:} Compute $z_m, \forall m$ using \eqref{eq:ratioalloc}
            \label{alginit:startaloc}\State Allocate one RF chain to each ER, $\mathcal{R}_m = \{m\}, \forall m$, $\bar{\mathcal{R}} = \{1, \ldots, M\}$, $w_m = 0, \forall m$
            \State $N'_v = N_v - M$, $N'_{rf, m} = \lceil z_m N'_v \rceil$, $\mathcal{R}' = \{\}$
            \Repeat
                \State \hspace{-2mm} $m^\star = \argmax_{m \notin \mathcal{R}'_m} z_m$, $\mathcal{R}' \leftarrow \mathcal{R}' \cup m^\star$, $i_c = M + 1$
                \Repeat
                    \State $\mathcal{R}_{m^\star} \leftarrow \mathcal{R}_{m^\star} \cup i_c$, $i_c \leftarrow i_c + 1$
                    \State $N'_{rf, m} \leftarrow  N'_{rf, m} - 1$, $N'_v \leftarrow N'_v - 1$
                \Until{$N'_{rf, m} = 0$ or $N'_v = 0$}
            \Until{$N'_v = 0$ or $|\mathcal{R}'| = M$}\label{alginit:endaloc}
            \For{$m = 1, \ldots, M$}\label{alginit:startpow}
                \State \hspace{-2mm} Compute $q_{i,l}, \forall i \in \mathcal{R}_m,  l$ using \eqref{eq:qinit} and \eqref{eq:lorentzweight}, $w_m = \tau_s$
                \State \hspace{-2mm} Solve \eqref{opt:low-comp_phase} to obtain $\hat{\omega_{i,n}}^\star, \forall i \in \mathcal{R}_m, n$ 
                \Repeat
                    \State $\bar{\omega}_{i, n} = w_m, \omega_{i, n} = \bar{\omega}_{i, n}e^{j\hat{\omega}_{i, n}} \forall i \in \mathcal{R}_m ,n$
                    \State Compute $P_{dc, m}$ using \eqref{eq:vout}, \eqref{eq:Pdc}, \eqref{eq:Erewrit1}, and \eqref{eq:Erewrit2}
                    \State $w_m \leftarrow \varsigma w_m$
                \Until{$P_{dc, m} \geq \bar{P}_m$}
            \EndFor\label{alginit:endpow}
\end{algorithmic} 
\end{algorithm}

\section{Numerical Analysis}\label{result}

In this section, we provide numerical analysis of the system performance. We consider an indoor office with a transmitter located at the center of the ceiling. The operating frequency of the system is $f_1 = 5.18$ GHz, which matches the characteristics of the utilized rectifier model \cite{clerckx2018beneficial}. The spacing between the elements in the DMA is $\lambda_1/5$, while $\lambda_1/2$ is the distance between two consecutive microstrips. Meanwhile, the inter-element distance is $\lambda_1/2$ in the fully-digital architecture. Thus, $N_v = N_h = \lfloor {2L}/{\lambda_1} \rfloor$ for the fully-digital system, and $N_v = \lfloor {2L}/{\lambda_1} \rfloor, N_h = \lfloor {5L}/{\lambda_1} \rfloor$ for the DMA-assisted system \cite{near-field}. Note that  $L$ is the array length while the arrays are considered to be square-shaped. We set the optimization parameters $\tau_s = 10^{-3}$, $\varsigma = 5$, and $\upsilon = 10^{-6}$. Without loss of generality, we set $G = 1$ and $\bar{\eta} = \frac{\pi}{4}$\footnote{In practice, the HPA output power is larger than the input power because $G > 1$. However, the proposed framework applies to any values of $G$.}. Finally, the rectifier parameters are $v_t = 25$ mV and $\eta_0 = 1.05$ \cite{clerckx2018beneficial, SISOAllClreckx, jointWFandBFMIMOclreckx}.

We utilize the characteristics of the \textregistered Rogers RO4000 series ceramic laminate to calculate the propagation coefficients of the microstrips. Specifically, we calculate the attenuation and propagation coefficients of a RO400C LoPro with a thickness of 20.7 mil (0.5258 mm) using the formulation provided in \cite{MyEBDMA}, which gives $\alpha = 0.356$ m$^{-1}$ and $\beta = 202.19$ m$^{-1}$. Based on the rectifier circuit design and simulations in \cite{clerckx2018beneficial}, the rectifier circuit diode enters the breakdown region when the received RF power is approximately $ 100\ \mu$W for a continuous wave ($N_f = 1$). Moreover, it has been shown that the maximum RF-to-DC conversion efficiency is approximately $20 \%$ for the mentioned setup. Hence, we establish a minimum requirement of $\Tilde{P}_{dc} = 20\ \mu$W for DC harvested power. In the figures, FD refers to the fully-digital architecture. Moreover, $d$ represents the distance between the ER and the center of the transmitter. Note that since the utilized problem formulation in \eqref{probmain} is equivalent to maximizing the PTE, the showcased results for power consumption inherently capture the PTE of the system. Specifically, the PTE and power consumption are inversely proportional.

In order to showcase the power consumption ($P_c$) in the figures, we first get the solution, i.e., the metamaterials' frequency response and the signal weights, by leveraging the proposed optimization framework. Then, $P_c$ is computed by sampling  \eqref{eq:power_consumption} using the Nyquist sampling theorem. Specifically, we consider a bandwidth $\textrm{BW} = 10$ MHz and sub-carrier spacing $\Delta f = \textrm{BW}/N_f$. Then, the sampling rate is computed based on the maximum frequency, such that $f_s = 2f_{N_f}$ is the sampling rate,  $\Delta t = 1/f_s$ is the sampling interval, and the signal duration is $1$ ms.

Fig.~\ref{fig:converge} provides the convergence performance of the proposed ASCA-DMA approach by presenting the power consumption at the end of each iteration of alternating optimization. It is seen that the objective value gradually decreases with the iterations until convergence, while the number of required iterations for convergence depends on the setup. For instance, it is shown that increasing $L$, $N_f$, or $M$ can increase the complexity of the problem, leading to more iterations. Fig.~\ref{fig:convergeall} provides a detailed algorithm performance by showing the results and the end of each iteration of optimization (both alternating and SCA), including the SCA algorithm for both optimizing metamaterials and digital precoders. It is seen that since the minimum harvested power increases when optimizing $\mathbf{q}$, the power consumption increases, while this facilitates decreasing the power consumption when optimizing the digital precoders. Therefore, the power consumption decreases gradually at the end of each iteration of the alternating optimization (red arrows) until convergence.

\begin{figure}[t]
    \centering
    \includegraphics[width=0.9\columnwidth]{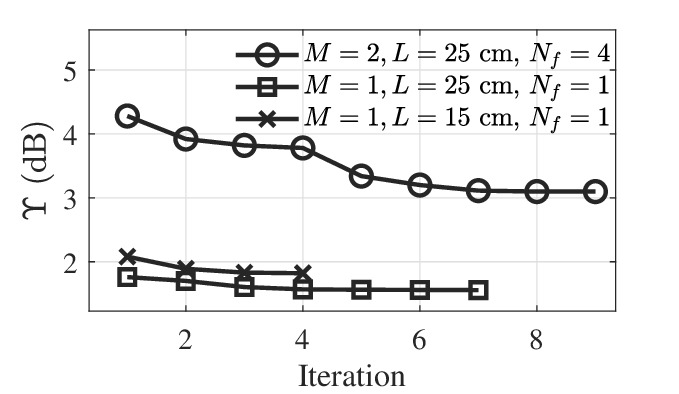}
    \caption{The convergence performance of ASCA-DMA for different $L$, $M$, and $N_f$ over iterations with $d = 2.2$ m and $P_{max} = 1$ W.}
    \label{fig:converge}
\end{figure}

\begin{figure}[t]
    \centering
    \includegraphics[width=0.8\columnwidth]{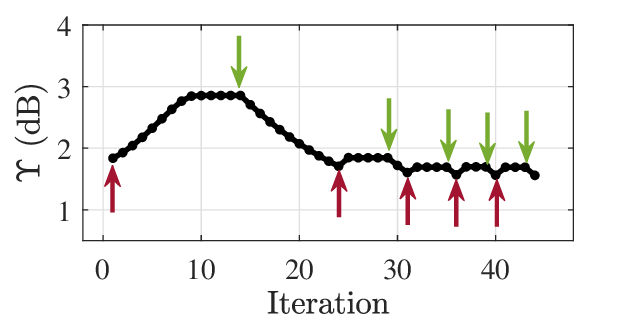}
    \caption{The convergence performance of ASCA-DMA for $M = 1$, $L = 25$ cm, $d = 2.2$ m, $N_f = 1$, and $P_{max} = 1$~W over all iterations (alternating and SCA). The red and green arrows indicate the starting points of optimization for the metamaterial elements' frequency response and digital precoders, respectively.}
    \label{fig:convergeall}
\end{figure}

Fig.~\ref{fig:overL} showcases the power consumption of the system as a function of the antenna length. Note that increasing $L$ reduces the power consumption since the number of elements and the array aperture increases. This leads to a better beam focusing capability, thus, delivering more power to the ERs given the same transmit power. Although DMA outperforms FD, the performance gap between DMA and FD highly depends on the system setup. For instance, the gap is larger for higher $P_{max}$ due to the larger number of RF chains and HPAs in FD. Moreover, observe that the performance gains caused by increasing $L$ depend also on the distance of the receiver from the transmitter. Specifically, when the distance becomes larger, the power consumption decreases more substantially by utilizing larger antenna arrays.

\begin{figure}[t]
    \centering
    \includegraphics[width=0.9\columnwidth]{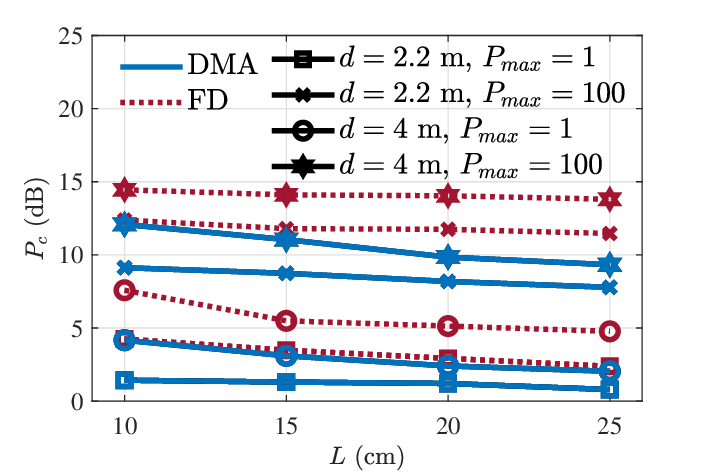}
    \caption{The power consumption as a function of $L$ for $N_f = 8$, $P_{max} \in\{1, 100\}$ W, and $d \in\{2.2, 4\}$ m.}
    \label{fig:overL}
\end{figure}

Our results in Fig.~\ref{fig:overN} corroborate that increasing the number of transmit tones reduces the power consumption. As discussed in Section~\ref{Section:Sys:TxRXSig}, this is because the HPAs operate in the linear regime, and higher $N_f$ can leverage the rectifier non-linearity and deliver more DC power to the ERs via waveform optimization. However, as previously mentioned, the DMA architecture requires less power consumption, but the performance gap between DMA and FD highly depends on the system parameters. Meanwhile, observe that increasing $N_f$ leads to more performance gains in FD compared to DMA since the number of signals/RF chains is relatively larger in FD, and this provides additional degrees of freedom in designing the transmit signal. Thus, the performance gap between DMA and FD becomes smaller as $N_f$ increases.

\begin{figure}[t]
    \centering
    \includegraphics[width=0.8\columnwidth]{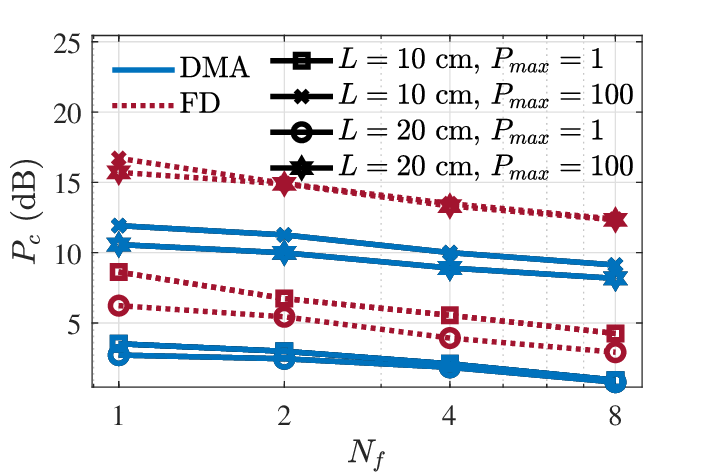}
    \caption{The power consumption as a function of $N_f$ for $P_{max} \in\{1, 100\}$~W, $L \in\{10, 20\}$ cm, and $d = 2.2$~m.}
    \label{fig:overN}
\end{figure}

Fig.~\ref{fig:overd} illustrates the impact of the distance $d$ of the ER from the transmitter on the system performance. It is obvious that the power consumption increases with the distance since the path loss becomes larger and more transmit power is required to overcome it and meet the EH requirements. Interestingly, it can be seen that the performance gains caused by increasing the antenna length are more substantial when the receiver is located at larger distances.

\begin{figure}[t]
    \centering
    \includegraphics[width=0.9\columnwidth]{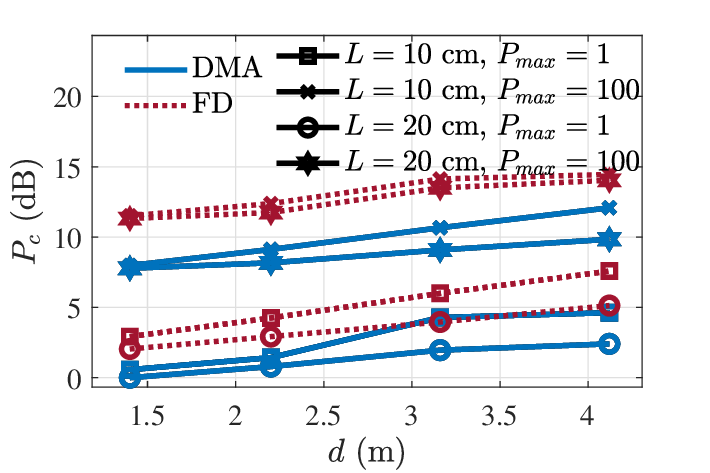}
    \caption{The power consumption as a function of ER distance for $N_f = 8$, $P_{max} \in\{1, 100\}$ W, and $L \in\{10, 20\}$ cm.}
    \label{fig:overd}
\end{figure}

The impact of the number of ERs on system performance is illustrated in Fig.~\ref{fig:overM} for different system parameters. As expected, the power consumption increases with $M$ since more EH requirements must be met, leading to more required transmit power. As previously shown, DMA outperforms FD, but the performance gap between these architectures varies with the system setup and parameters such as antenna length and HPAs' saturation power.

\begin{figure}[t]
    \centering
    \includegraphics[width=0.85\columnwidth]{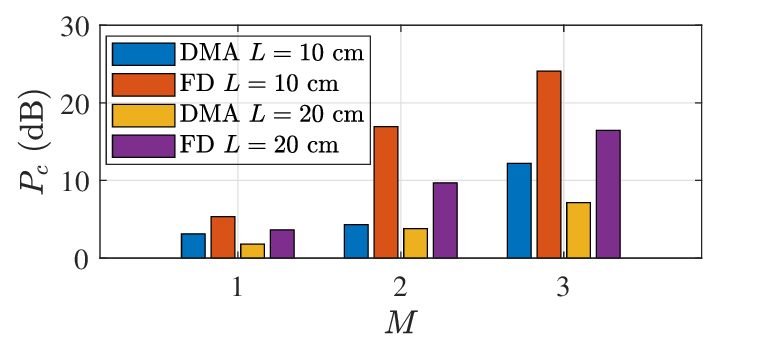}
    \includegraphics[width=0.85\columnwidth]{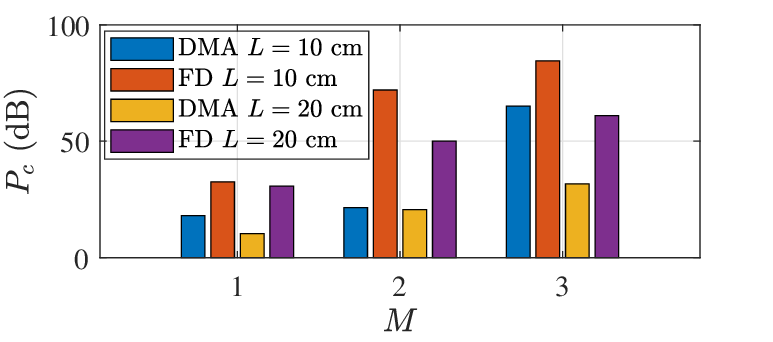}
    \caption{The power consumption over $M$ for (a) $P_{max} = 1$~W (top) and (b) $P_{max} = 100$ W (bottom), while $N_f = 4$, and $L \in\{10, 20\}$~cm. The ERs are located at $d = 3.1$~m.}
    \label{fig:overM}
\end{figure}

Fig.~\ref{fig:field_analysis} provides some insights regarding the beam focusing capability in the near-field WPT by illustrating the normalized received RF power in each spatial point of the area. Note that the received signal at each point is normalized by its path loss to remove the impact of the distance. In Fig.~\ref{fig:field_analysis}a, it is seen that when the ER is located in the near-field region, the beam is focused around its location, while the beam trace fades increasingly past the ER. Such phenomena can have a huge benefit in reducing the RF emission footprint in the environment, which facilitates the implementation of environmentally friendly WPT systems \cite{lópez2023highpower}. Meanwhile, the beam pattern for a located device in the far-field is formed in the receiver's direction, as illustrated in Fig.~\ref{fig:field_analysis}b. This may be highly disadvantaged in interference-sensitive applications since the generated beam may cause difficult-to-handle interference in the signals conveying information, e.g., in SWIPT.

\begin{figure}[t]
\centering
    \includegraphics[width=0.49\columnwidth]{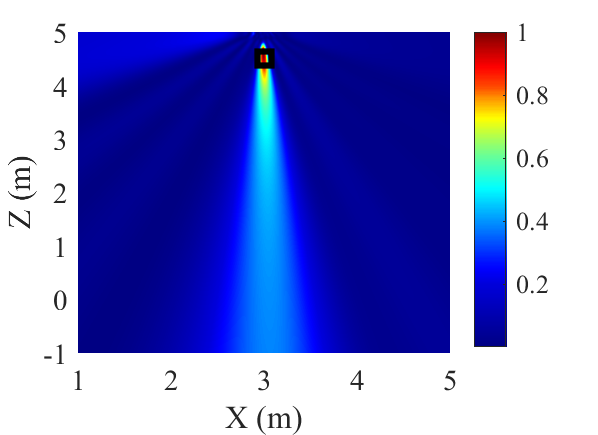} 
    \includegraphics[width=0.49\columnwidth]{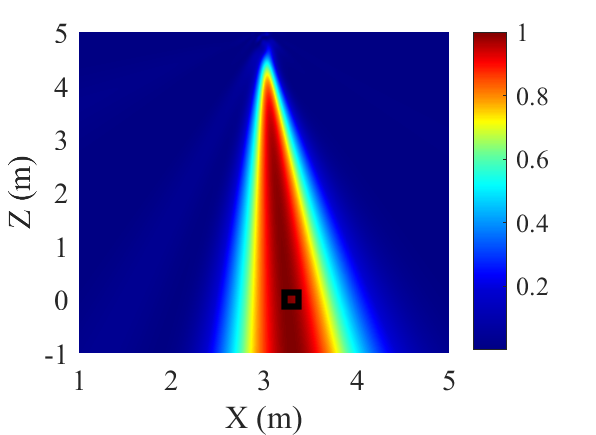}
    \caption{The normalized received RF power (W) in the area when the energy receiver is located at (a) the near-field region (left) and (b) the far-field region (right) in the DMA-assisted system with $L = 30$~cm and $N_f = 1$.}
    \label{fig:field_analysis}
\end{figure}

Fig.~\ref{fig:initcomp} compares the performance of the proposed initialization method with a random initialization approach. For the latter, the variables are generated randomly 1000 times, and the set of variables that leads to the least power consumption while satisfying the EH requirements is chosen as the initial point of the optimization. Moreover, if the EH constraints are not met within 1000 realizations, the random generation is continued until the EH requirements are met. It is seen that the initialization method proposed in Section~IV is effective and can lead the optimization framework to provide significantly better results in both transmit architectures.
\begin{figure}[t]
    \centering
    \includegraphics[width=0.9\columnwidth]{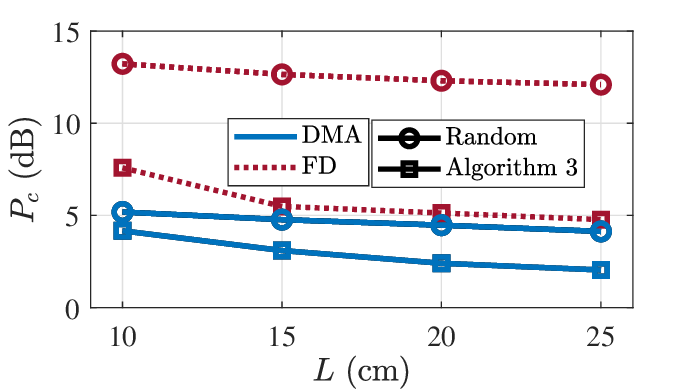}
    \caption{The power consumption using different initialization methods in DMA and fully-digital architectures as a function of $L$ for $M = 1$, $P_{max} = 1$~W, $d = 2.2$ m, and $N_f = 8$.}
    \label{fig:initcomp}
\end{figure}

The received RF waveform for both transmit architectures is presented in Fig.\ref{fig:waveform}.a. As previously mentioned, high PAPR waveforms are beneficial for enhancing the performance in terms of DC harvested power \cite{papr2}, and our simulations verify this by showing that the received signal experiences high peak amplitudes at specific intervals. Note that the peak-to-peak time depends on the characteristics of the EH circuit, mainly the capacitor. Recall that when HPAs are operating in the linear regime, as in our case, the HPA does not introduce distortion to the signal. Thus, it is beneficial to utilize multiple tones to leverage the rectifier's non-linearity. Moreover, the average amplitude of the frequency components among the RF chains and the sub-carrier is showcased in Fig.\ref{fig:waveform}.b. It is seen that all frequency components are utilized and power is allocated to them with different ratios. Moreover, observe that the transmit power is distributed among all RF chains in both architectures to efficiently leverage the degrees of freedom in focusing the signal toward the receiver and meeting the EH requirement with less power consumption.

\begin{figure}[t]
\centering
    \includegraphics[width=0.49\columnwidth]{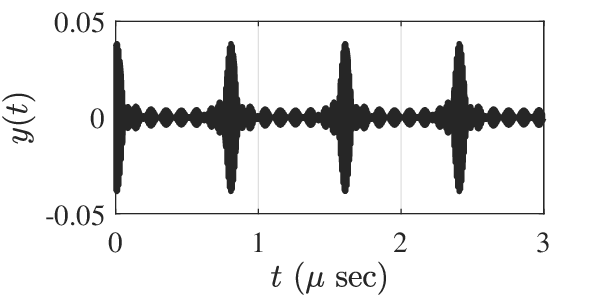} 
    \includegraphics[width=0.49\columnwidth]{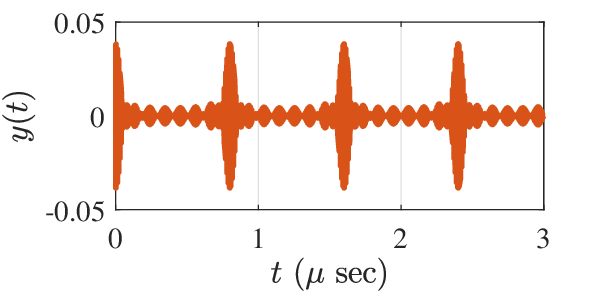} \
    \includegraphics[width=0.49\columnwidth]{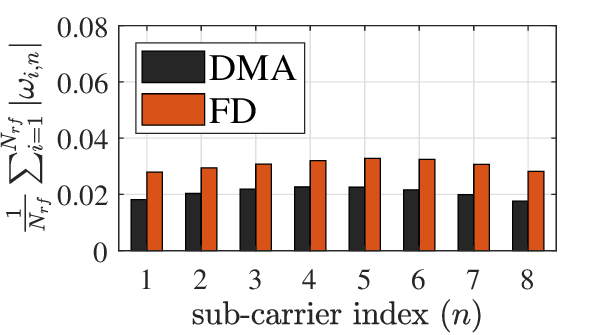} 
    \includegraphics[width=0.49\columnwidth]{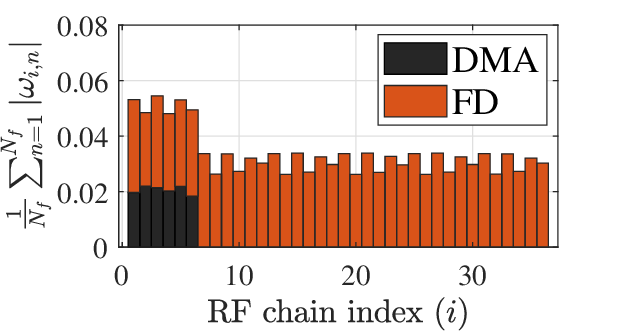}
    \caption{(a) The received signal with DMA (top-left) and FD (top-right), and (b) the average of the amplitude of the frequency components of the RF chains output signals over the RF chains (bottom-left) and over the sub-carriers (bottom-right) for $M = 1$, $L = 20$ cm, and $N_f = 8$.}
    \label{fig:waveform}
\end{figure}

Some discussions on the practical challenges and complexities of implementing both architectures are in order. Notice that many more RF chains are needed in an FD array compared to a DMA with the same size, as illustrated in Fig.~\ref{fig:elementnum}. For instance, 64 RF chains are needed in a 25 $\times$ 25 cm$^2$ FD array and just 9 for a DMA with the same size at $f_1 = 5.18$ GHz. Notably, as $L$ substantially increases, the number of required RF chains for an FD array increases drastically as well, making DMA the reasonable implementation option. However, this must be taken with caution as studies considering mutual coupling are necessary. Such practical considerations may impact the performance of the DMA, especially when the inter-element distance becomes smaller. On the other hand, both architectures might lose some performance gains due to the sub-optimality of the algorithm. Thus, the actual performance gaps and trade-offs might change, and obtaining optimal solutions for this problem is challenging and still an open topic for research.

\begin{figure}[t]
    \centering
    \includegraphics[width=0.45\columnwidth]{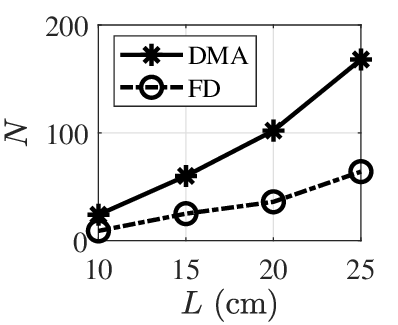}
    \includegraphics[width=0.45\columnwidth]{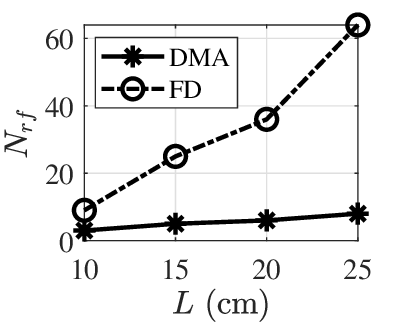}
    \caption{The total number of (a) antenna elements (left) and (b) RF chains (right) as a function of antenna length $L$.}
    \label{fig:elementnum}
\end{figure}

\section{Conclusion and Future Work}\label{conclusion}

In this paper, we investigated a multi-antenna near-field WPT system with a DMA-assisted transmitter to charge multiple non-linear EH devices. Furthermore, we proposed an optimization framework relying on alternating optimization and SCA for the joint waveform optimization and beam focusing to minimize the system power consumption while meeting ERs' EH requirements. Numerical results showed that DMA-assisted architecture outperforms the fully-digital structure in terms of power consumption, while their performance gap depends on the system setups and parameters such as antenna length, saturation power of the HPAs, number of ERs, and ER distance to the transmitter. Moreover, we showed that increasing the antenna length or the number of tones can enhance the performance. Finally, we verified that the transmitter can focus the energy beams on the spatial points in the near-field region, while energy beams are formed toward the devices' direction in the far-field.

As a prospect for future research, we may delve deeper into the signal generation aspect by analyzing the power consumption based on the number of tones. Another research direction is to utilize optimization approaches with lower complexity, e.g., relying on machine learning, to learn online the input-output relation of the system's non-linear components while optimizing the transmit waveform accordingly. Moreover, we may consider a mutual coupling-aware optimization to investigate the impact of the coupling between the antenna elements in the DMA architecture, which requires novel solutions and might change the system performance depending on the inter-element distance. Additionally, the proposed frameworks can be extended to traditional hybrid and analog architectures consisting of phase shifters to investigate the power consumption of those architectures compared to fully-digital and DMA-assisted structures. Finally, proposing novel solutions for DMA-assisted SWIPT systems with non-linear energy harvesters can be another interesting research opportunity.

\appendix

\subsection{Proof of Theorem~\ref{theorem:1}}\label{appen1}

We can write 
\begin{align}
    &\sum_{i = 1}^{N_{rf}} \mathbb{E}_t\biggl\{\sqrt{\sum_{l = 1}^{N_h} x^{DMA}_{i,l}(t)^2}\biggr\} \labelrel\leq{reformWfix1:1} \sum_{i = 1}^{N_{rf}} \sqrt{\mathbb{E}_t\biggl\{\sum_{l = 1}^{N_h} x^{DMA}_{i,l}(t)^2\biggr\}} \nonumber \\ &\labelrel={reformWfix1:2} \sum_{i = 1}^{N_{rf}} \sqrt{\sum_{l = 1}^{N_h} \mathbb{E}_t\biggl\{x^{DMA}_{i,l} (t)^2\biggr\}} \labelrel={reformWfix1:3} \sum_{i = 1}^{N_{rf}} \sqrt{\sum_{l = 1}^{N_h} \sum_{n = 1}^{N_f}\frac{G^2}{2}|\omega_{i,n} q_{i,l}h_{i,l}|^2},
\end{align}
where \eqref{reformWfix1:1} comes from the Jensen inequality \cite{mcshane1937jensen}, \eqref{reformWfix1:2} from the linearity of the mathematical average operator, and \eqref{reformWfix1:3} from the equivalence between the average power of the signal in time and frequency domain. Furthermore, we can write 
\begin{equation}
    \sum_{i = 1}^{N_{rf}} \sqrt{\sum_{l = 1}^{N_h} \sum_{n = 1}^{N_f}|\omega_{i,n} q_{i,l}h_{i,l}|^2} \!=\! \sum_{i = 1}^{N_{rf}} \sqrt{\sum_{l = 1}^{N_h} |q_{i,l}h_{i,l}|^2} \sqrt{\sum_{n = 1}^{N_f}|\omega_{i,n}|^2},
\end{equation}
where $\sqrt{\sum_{n = 1}^{N_f}|\omega_{i,n}|^2}$ represents the $l_2$ norm operator for the vector $[\omega_{i,1}, \ldots, \omega_{i, N_{rf}}]^T$, which is convex, leading to the convexity of the upper bound for fixed $q_{i,l}$.

\subsection{Proof of Theorem~\ref{theorem:2}}\label{appen2}

We relax \eqref{probdmaQg} by limiting the values of $q_{i,l}$ to lie within the Lorentzian circle in the complex plane. By utilizing the fact that \eqref{probdmaQcona} is a positive and increasing function of the rectifier's output voltage and using the epigraph form, the relaxed problem can be written as \eqref{probdmaQQcon}. Note that each configuration of the frequency response of a metamaterial corresponds to a point on the Lorentzian-constrained circle. Imagine $\vec{e}$ is the vector that represents the direction and gain of a point on the Lorentzian circle, while this direction and gain impacts the transmit signal. Meanwhile, the goal of \eqref{probdmaQQcon} is to increase the minimum output voltage of the ERs. Therefore, when shaping the transmit signal toward different ERs, it is obvious that $\vec{e}$ should be chosen with the maximum gain possible in the required direction to improve the signal strength at the receiver. Furthermore, the maximum feasible gain introduced by a metamaterial element along a specified direction happens when the point is exactly on the Lorentzian-constrained circle in that direction. Hence, although \eqref{probdmaQQQg} is a relaxed version of the constraint \eqref{probdmaQg}, considering the final solution of the metamaterials as $a\vec{e},\ 0 \leq a \leq 1$, the only solution that leads to the maximum gain for the desired direction is $a = 1$. Thus, the solution of \eqref{probdmaQQcon} is the same as \eqref{probdmaQcon} leading to a configuration positioned on the Lorentzian circle and the equivalence between the problems.

\ifCLASSOPTIONcaptionsoff
  \newpage
\fi

\bibliography{ref_abbv}
\bibliographystyle{ieeetr}

\begin{IEEEbiography}[{\includegraphics[width=1in,height=1.25in,clip,keepaspectratio]{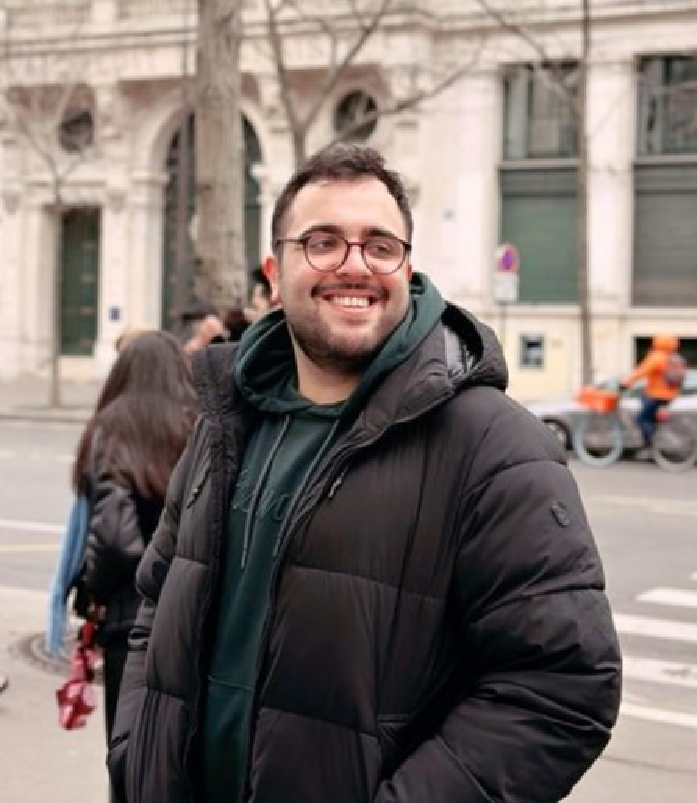}}]{Amirhossein Azarbahram}~(Graduate Student~Member, IEEE) received the B.Sc. degree in telecommunications from KN Toosi University of Technology, Iran, in 2020, the M.Sc. degree in communications system from Sharif University of Technology, Iran, in 2022. He is currently a Doctoral researcher at the Centre for Wireless Communications (CWC), University of Oulu, Finland. He was a visiting researcher at the Connectivity Section (CNT) of the Department of Electronic Systems, Aalborg University (AAU), Denmark, in 2024. His research interests include wireless communications and RF wireless power transfer. 
\end{IEEEbiography}

\begin{IEEEbiography}[{\includegraphics[width=1in,height=1.25in,clip,keepaspectratio]{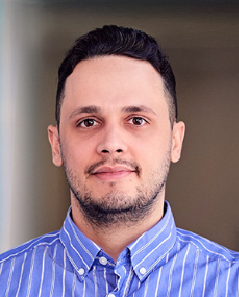}}]{Onel L. A. López}~(Senior~Member, IEEE) received the B.Sc. (1st class honors, 2013), M.Sc. (2017), and D.Sc. (with distinction, 2020) degree in Electrical Engineering from the Central University of Las Villas (Cuba), the Federal University of Paraná (Brazil), and the University of Oulu (Finland), respectively. From 2013-2015, he served as a specialist in telematics at the Cuban telecommunications company (ETECSA). He is a collaborator to the 2016 Research Award given by the Cuban Academy of Sciences, a co-recipient of the 2019 and 2023 IEEE European Conference on Networks and Communications (EuCNC) Best Student Paper Award, and the recipient of both the 2020 best doctoral thesis award granted by Academic Engineers and Architects in Finland TEK and Tekniska Föreningen i Finland TFiF in 2021 and the 2022 Young Researcher Award in the field of technology in Finland. He is co-author of the books entitled ``Wireless RF Energy Transfer in the massive IoT era: towards sustainable zero-energy networks'', Wiley, 2021, and ``Ultra-Reliable Low-Latency Communications: Foundations, Enablers, System Design, and Evolution Towards 6G'', Now Publishers, 2023. Since 2024, he has been an Associate Editor of the IEEE Transactions on Communications and IEEE Wireless Communications Letters. He is currently an Associate Professor (tenure track) in sustainable wireless communications engineering at the Centre for Wireless Communications (CWC), Oulu, Finland. His research interests include sustainable IoT, energy harvesting, wireless RF energy transfer, wireless connectivity, machine-type communications, and cellular-enabled positioning systems.
\end{IEEEbiography}

\begin{IEEEbiography}[{\includegraphics[width=1in,height=1.25in,clip,keepaspectratio]{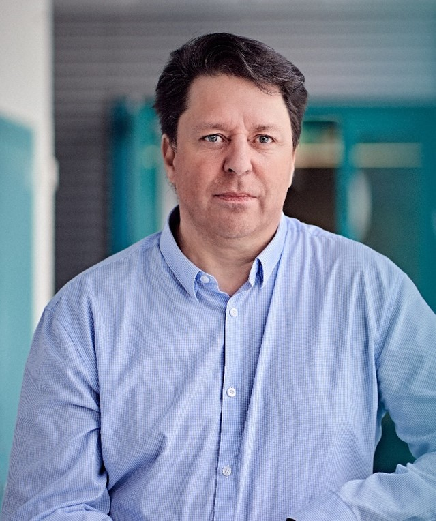}}]{Matti~Latva-aho}~(Fellow, IEEE) received his M.Sc., Lic.Tech., and Dr.Tech. (Hons.) degrees in Electrical Engineering from the University of Oulu, Finland, in 1992, 1996, and 1998, respectively. From 1992 to 1993, he was a Research Engineer at Nokia Mobile Phones in Oulu, Finland, after which he joined the Centre for Wireless Communications (CWC) at the University of Oulu. Prof. Latva-aho served as Director of CWC from 1998 to 2006 and was Head of the Department of Communication Engineering until August 2014. He is currently a Professor of Wireless Communications at the University of Oulu and the Director of the National 6G Flagship Programme. He is also a Global Fellow at The University of Tokyo. Prof. Latva-aho has published over 500 conference and journal papers in the field of wireless communications. In 2015, he received the Nokia Foundation Award for his achievements in mobile communications research.
\end{IEEEbiography}



%\newpage


%\vfill


\end{document}